\documentclass[12pt,centertags,reqno]{amsart}

\usepackage[foot]{amsaddr}
\usepackage{mathrsfs, mathtools}
\usepackage{latexsym}
\usepackage[english]{babel}
\usepackage[T1]{fontenc}

\usepackage[utf8]{inputenc}
\usepackage{amssymb}
\usepackage[numbers]{natbib}
\usepackage{url}
\usepackage{hyperref}
\usepackage{verbatim}
\usepackage{graphicx}
\usepackage{leftidx}
\mathtoolsset{showonlyrefs}
\usepackage{booktabs}
\usepackage{caption}

\numberwithin{equation}{section} \makeatletter
\renewcommand{\subsection}{\@startsection
{subsection}{2}{0mm}{\baselineskip}{-0.25cm}
{\normalfont\normalsize\bf}} \makeatother

\newtheorem{theorem}{Theorem}[section]
\newtheorem{lemma}[theorem]{Lemma}
\newtheorem{corollary}[theorem]{Corollary}
\newtheorem{remark}[theorem]{Remark}

\def \F {\mathcal F}

\def \N {\mathcal N}
\def \L {\mathcal L}
\def \P {\mathbf P}
\def \Q {\mathbf Q}
\def \R {\mathbb R}
\def \E {\mathcal E}
\def \I {{\mathbf 1}}
\def \bF {\mathbb F}

\def \bN {\mathbb N}

\newcommand*\xbar[1]{%
  \hbox{%
    \vbox{%
      \hrule height 0.5pt % The actual bar
      \kern0.5ex%         % Distance between bar and symbol
      \hbox{%
        \kern-0.1em%      % Shortening on the left side
        \ensuremath{#1}%
        \kern-0.1em%      % Shortening on the right side
      }%
    }%
  }%
}

\newcommand{\ud}{\mathrm d}
\newcommand{\ds}{\displaystyle}
\newcommand{\esp}[2][\mathbb E] {#1\left[#2\right]}
\newcommand{\espp}[2][\mathbb E^{\widehat \P}] {#1\left[#2\right]}
\newcommand{\espbar}[2][\mathbb E^{\bar \P}] {#1\left[#2\right]}

\newcommand{\doleans}[1] {\mathcal E\left(#1\right)}

\begin{document}

\date{}

\author[A.~Cretarola]{Alessandra Cretarola}
\address{Alessandra Cretarola, Department of Mathematics and Computer Science,
 University of Perugia, Via Luigi Vanvitelli, 1, I-06123 Perugia, Italy.}\email{alessandra.cretarola@unipg.it}

\author[G.~Figà-Talamanca]{Gianna Figà-Talamanca}
\address{Gianna Figà Talamanca, Department of Economics,
University of Perugia, Via Alessandro Pascoli,
I-06123 Perugia, Italy.}\email{gianna.figatalamanca@unipg.it}

\author[M. Patacca]{Marco Patacca}
\address{Marco Patacca, Department of Economics,
University of Perugia, Via Alessandro Pascoli,
I-06123 Perugia, Italy.}\email{marco.patacca@studenti.unipg.it}

\date{This is an updated version (September 2017) of the paper \textit{A confidence-based model for asset and derivatives pricing in the BitCoin market}}

%\title[Modelling BitCoins via market reputation and option pricing
%]{Modelling BitCoins via market reputation\\ and option pricing}

%
%
\title[A sentiment-based model for the BitCoin: theory, estimation and option pricing
]{A sentiment-based model for the BitCoin: theory, estimation and option pricing%\footnote{This is an updated version of the paper \textit{A confidence-based model for asset and derivative prices in the BitCoin market}.}
}

\begin{abstract}
In recent literature it is claimed that BitCoin price behaves more likely to a volatile stock asset than a currency and that changes in its price are influenced by sentiment about the BitCoin system itself; in \citet{MainDrivers} the author analyses transaction based as well as popularity based potential drivers of the  BitCoin price finding positive evidence. Here, we endorse this finding and consider a bivariate model in continuous time to describe the price dynamics of one BitCoin as well as a second factor, affecting the price itself, which represents a sentiment indicator. We prove that the suggested model is arbitrage-free under a mild condition and, based on risk-neutral evaluation,  we obtain a closed formula to approximate the price of European style derivatives on the BitCoin. By applying the same approximation technique to the joint likelihood of a discrete sample of the bivariate process, we are also able to fit the model to market data. This is done by using both the \textit{Volume} and the number of \textit{Google searches} as possible proxies for the sentiment factor. Further, the performance of the pricing formula is assessed on a sample of market option prices obtained by the website deribit.com.
%
%In \cite{bukovina2016sentiment} the BitCoin price is related to measures of sentiment on the BitCoin system  based on Natural Language Processing techniques and provided by the website Sentdex.com 
% A possible delay between the sentiment measure and its delivered effect on the BitCoin price is also considered. Note that in our approach sentiment is modeled as a continuous non negative value so it is not consistent with the approach in \cite{bukovina2016sentiment} where Natural Language Processing techniques lead to a binary valued sentiment (positive or negative). 
%From a theoretical viewpoint we give three contributions: the model is proven to be arbitrage-free under proper conditions and its statistical properties are investigated. Further, the likelihood for a discrete sample of the model is computed and an approximated closed formula is derived so that maximum likelihood estimates can be obtained for model parameters once the delay is known; otherwise the profile likelihood gives a feasible approach to fit the model to real data. Then, based on risk-neutral evaluation, a quasi-closed formula is derived for European style derivatives on the BitCoin. 
\end{abstract}

\maketitle


%{\bf JEL Classification}: C10,C13,C22,C32,C58,C63,E44,E47,G12,G13,G17,G41.

%{\bf AMS Classification}: 60Gxx, 62Fxx,62Mxx.

\section{Introduction}

The BitCoin was first introduced as an electronic payment system between peers. It is based on an open source software which generates a peer to peer network. This network includes a high number of computers connected to each other through the Internet and complex mathematical procedures are implemented both to check the truthfulness of the transaction and to generate new BitCoins. Opposite to traditional transactions, which are based on the trust in financial intermediaries, this system relies on the network, on the fixed rules and on cryptography. The open source software was created in 2009 by a computer scientist known under the pseudonym Satoshi Nakamoto, whose identity is still unknown. BitCoin has several attractive properties for consumers: it does not rely on central banks to regulate the money supply and it enables essentially anonymous transactions.  BitCoins can be purchased on appropriate websites that allow to change usual currencies in BitCoins. Further, payments can be made in BitCoins for several online services and goods and its use is increasing. Special applications have been designed for smartphones and tablets for transactions in BitCoins and some ATM have appeared all over the world (see Coin ATM radar) to change traditional currencies in BitCoins. At very low expenses it is also possible to send cryptocurrency internationally. However, the downside of BitCoin is that, due to anonymous transactions, it has been labeled as an exchange for organized crime and money laundering. It is worth to mention the recent Wannacry malware which last May has infected the informatic systems of many huge companies as well as thousands of computers around the world. The hackers which have spread this malware asked a ransom of 300 to 600 USD to be payed in BitCoins in order to get each computer rid of the infection.  
% as shown in Figure \ref{BitCoin}.
%\begin{figure}[htbp]
%\includegraphics[width=13cm, height=6cm]{%D:/%LavoriRicerca/InCorso/BitCoin/Nostro/FigurePaper/
%BitCoin.jpg}
%\caption{Daily average BitCoin price (Dec 2012 - Sep 2016).} \label{BitCoin}
%\end{figure}
Besides, BitCoins can be only deposited in a digital wallet which is costly and possibly subject to hacking attacks, thefts and other issues related to cyber-security. 
In spite of all the above critics, BitCoin has experienced a rapid growth both in value and in the number of transactions. A number of competitors, so called alt-coins, have also appeared recently without reaching the popularity of BitCoin; the most successful among these is Ethereum.
% A great interest has also spread on the potential of the Blockchain technology which lies behind each transaction in BitCoins.
Economic and financial aspects of BitCoin have been frequently addressed by financial blogs and by financial media but, until recently, researchers in Academia were primarily focused on the underlying technology and on safety and legal issues such as double spending.  
From an economic viewpoint, one of the main concerns about BitCoin is whether it should be considered a currency, a commodity or a stock. In \citet{Yermack}, the author performs a detailed qualitative analysis of BitCoin behavior. He remarks that a currency is usually characterized by three properties: a medium of exchange, a unit of account and a store of value. BitCoin is indeed a medium of exchange, though limited in relative volume of transactions and essentially restricted to online markets; however it lacks the other two properties. BitCoin value is rather volatile and traded for different prices in different exchanges, making it unreliable as a unit of account.   The conclusion in  \citet{Yermack} is that BitCoin behaves as a high volatility stock and that most transactions on BitCoins are aimed to speculative investments. 
%A second issue about BitCoin prices is the possibility of arbitrage given that it is traded on different web-exchanges for different prices; a pioneering theoretical contribution on this topic is given by \citet{doi:10.1080/14697688.2016.1231928}.
In recent years several papers have also appeared in order to analyze which are the main drivers of its price evolution in time; many authors claim that the high volatility in BitCoin prices may depend on sentiment and popularity about the BitCoin market itself; of course sentiment and popularity on BitCoin are not directly observed  but several variables may be considered as indicators, from the more traditional volume or number of transactions to the number of Google searches or Wikipedia requests about the topic, in the period under investigation. Main references in this area are \citet{MainDrivers, GoogleTrends, SentimentAnalysis}. Alternatively, in \citet{bukovina2016sentiment} a sentiment measure related to the BitCoin system is obtained from the website Sentdex.com. This website collect data on sentiment through an algorithm, based on Natural Language Processing techniques,  which is capable of identifying string of words conveying positive, neutral or negative sentiment on a topic (BitCoin in this case). The authors of the paper develop a model in discrete time and show that excessive confidence on the system may boost a Bubble on the BitCoin price.  
Motivated by the evidences in the above quoted papers we introduce a bivariate model in continuous time to describe both the dynamics of a BitCoin sentiment indicator and of the corresponding BitCoin price.
From the theoretical viewpoint we give three contributions: the model is proven to be arbitrage-free under proper conditions and its statistical properties are investigated. Then, based on risk-neutral evaluation, a quasi-closed formula is derived for any European style derivative on the BitCoin. It is worth noticing that a market for derivatives on BitCoin has recently raised on appropriate websites such as  \emph{https://coinut.com} and \emph{https://deribit.com} trading European Calls and Puts as well as Binary option endorsing the idea in \citet{Yermack} that BitCoins are likely to be used for speculative purposes.
Further, the likelihood for a discrete sample of the model is computed and an approximated closed formula is derived so that maximum likelihood estimates can be obtained for model parameters. Precisely, we suggest a two-step maximum likelihood method, the profile likelihood described in \citet{Davison:statmod, Pawitan:allLik} to fit the model to market data. From the empirical viewpoint we contribute to the literature by fitting the suggested model to market data considering both the Volume and the number of Google searches as proxies for the sentiment factor.  Besides the performance of the pricing formula is assesses on a sample of market option prices obtained by the website \emph{https://deribit.com}.\\
The rest of the paper is structured as follows. In Section \ref{sec:model} we describe the model for the BitCoin price dynamics and show that the market is arbitrage-free under a mild condition. In Section \ref{sec:loglike} we compute the joint distribution of the discretely sampled model as well as a closed form approximation. Then, we propose a statistical estimation procedure based on the corresponding approximated profile likelihood. In Section \ref{sec:option-pricing} we prove a quasi-closed formula for European-style derivatives with detailed computations for Plain Vanilla and Binary option prices. Section \ref{sec:fit} is devoted to test the possible proxies for the sentiment indicator, such as the number and volume of BitCoin transactions or the internet searches on Google and Wikipedia and to apply the whole estimation procedure to market data obtained from \emph{http://blockchain.info}. In Section \ref{sec:assessing} we evaluate model performance for option pricing considering options traded on \emph{https://deribit.com} for some "test" days. Finally, in Section \ref{sec:remarks} we give some concluding remarks and draw directions for interesting future investigations. The Appendices collect a brief description of the Levy approximation approach and the proofs of most of the technical results.

\section{The BitCoin market model}\label{sec:model}

We fix a probability space $(\Omega,\F,\P)$ endowed with a filtration 
%$\bF = \{\F_t,\ t \in  [0,T]\}$ 
$\bF = \{\F_t,\ t \ge 0\}$ that satisfies the usual conditions of right-continuity and completeness. %where $T > 0$ is a fixed and finite time horizon; 
%Furthermore, we assume that $\F_\infty = \F$. % and that $\F_0$ coincides with the trivial sigma-algebra. 
%Let us assume the existence of 
On the given probability space, we consider a main market in which heterogeneous agents buy or sell BitCoins and denote by %$P_t$ the price at time $t$ 
$S = \{S_t,\ t \geq 0\}$ the price process
of the cryptocurrency. We assume that the BitCoin price dynamics is described by 
the following equation:
\begin{equation} \label{eq:S}
\ud S_t = \mu_S P_{t-\tau} S_t \ud t+\sigma_S \sqrt{P_{t-\tau}}S_t\ud W_t,\quad  S_0=s_0 \in \R_+,
\end{equation}
where $\mu_{S} \in \R \setminus \{0\}$, $\sigma_{S} \in \R_+,\ \tau \in \R_+$ represent model parameters;  
$W = \{W_t,\ t \ge 0\}$ is a standard $\bF$-Brownian motion on $(\Omega,\F,\P)$ and
$P = \{P_t,\ t \geq 0\}$ is a stochastic factor, representing the sentiment index in the BitCoin market, satisfying
\begin{equation}\label{eq:BSdyn}
\ud P_t =\mu_P P_t\ud t+\sigma_P P_t\ud Z_t, %
\quad P_t = \phi(t),\ t \in [-L,0].
\end{equation}
Here, $\mu_P \in \R \setminus \{0\}$, $\sigma_P \in \R_+$, $L \in \R_+$, $Z = \{Z_t,\ t \geq 0\}$ is a standard $\bF$-Brownian motion on $(\Omega,\F,\P)$, which is $\P$-independent of $W$,
%possibly correlated with $W$, so that $\ud \langle W,Z\rangle_t = \rho \ud t$, for some constant $\rho \in [-1,1]$, 
and $\phi:[-L,0] \to [0,+\infty)$ is a continuous (deterministic) initial function. Note that, the non negative property of the function $\phi$ corresponds to require that the minimum level for sentiment is zero.
It is worth noticing that in \eqref{eq:BSdyn} we also consider the effect of the past, since we assume that the sentiment factor $P$ affects explicitly the BitCoin price $S_t$ up to a certain preceding time $t-\tau$. Assuming that $\tau<L$ and that factor $P$ is observed in the period $[-L,0]$ makes the bivariate model jointly feasible.
It is well-known that the solution of \eqref{eq:BSdyn} is available in closed form and that $P_t$ has a lognormal distribution for each $t > 0$, see \citet{black1973pricing}. 
Here, $P$ stands for an exogenous factor affecting the instantaneous variance of the BitCoin price changes modulated by $\sigma_S$. It is worth noticing that the instantaneous variance of the BitCoin  price process increases with the delayed process $P$; this may appear counter-intuitive if $P$ is interpreted  only as a \emph{positive sentiment} indicator. However, in our perspective, the factor $P$ is  mathematically a non-negative variable but does not necessary represent a \emph{positive sentiment} indicator. Examples are the volume or number of transactions; these are non negative values which increase with both short (fear) and long (enthusiasm) positions in BitCoins. Similarly, the number of internet searches within a fixed time period cannot go negative but internet searches may increase both with enthusiasm and fear about the BitCoin System, or whatever other financial asset we would like to model with the dynamics we suggest here. Hence, an increase in $P$ may be actually related to an increase in the uncertainty about the the price $S$. 
In order to visualize the dynamics implied by the model in equations \eqref{eq:S} and \eqref{eq:BSdyn}, we plot in Figure \ref{CambioTau} a possible simulated path of daily observations for the sentiment factor $P$ and the corresponding BitCoin prices $S$ within one year horizon by letting $\tau$ vary; as expected, market reaction to sentiment is delayed when $\tau$ increases.

\begin{figure}[htbp]
\includegraphics[width=13cm, height=6cm]{%D:/LavoriRicerca/InCorso/BitCoin/Nostro/FigurePaper/
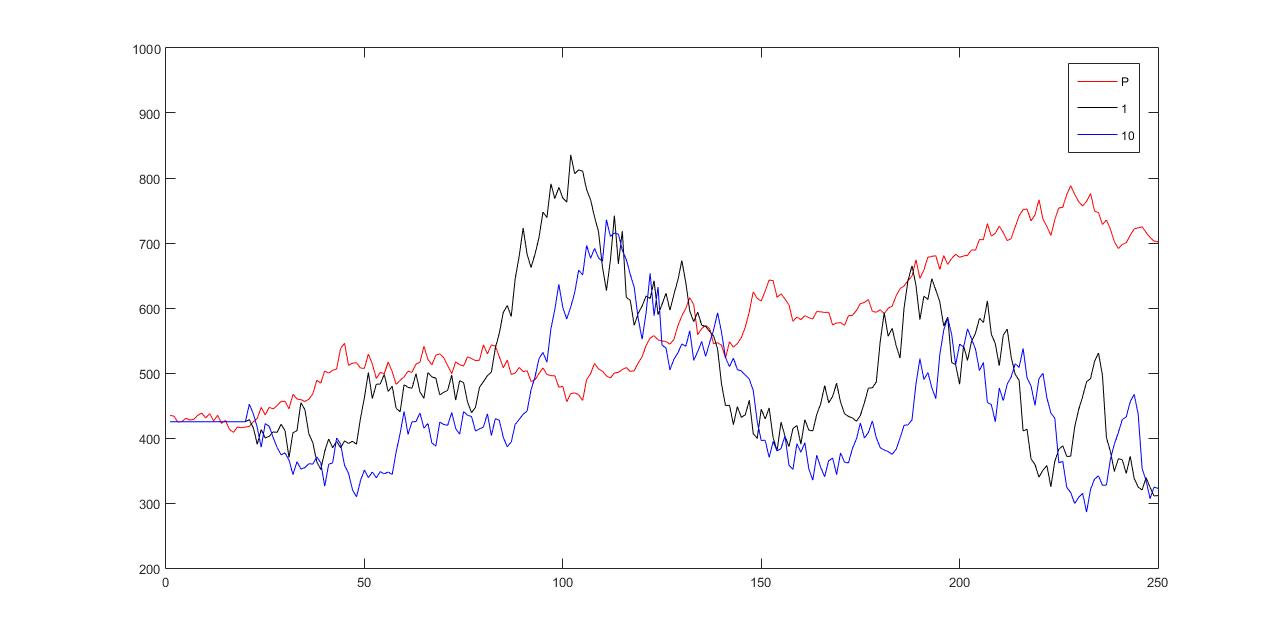} 
\caption{An example of BitCoin price dynamics given the evolution of the sentimente index (red): $\tau=1$ day (black), $\tau=10$ days (blue). Model parameters are set to $\mu_P=0.03,\sigma_P=0.35$, $\mu_S=10^{-5},\sigma_S=0.04$.}\label{CambioTau}
\end{figure}

\begin{remark}
The suggested model is motivated by the outcomes in \citet{MainDrivers, GoogleTrends, SentimentAnalysis,bukovina2016sentiment} where the authors relate the BitCoin price dynamics to some sentiment factors and does not take into account special features of the BitCoin market such as the underlying technology or the order mechanism. Besides, BitCoin is treated as a financial stock as suggested in \citet{Yermack}. Possibly, the model can be applied to any other financial assets whether one believes their price to depend on suitably identified sentiment indicators.
\end{remark}

We assume that the reference filtration $\bF=\{\F_t,\ t \geq 0\}$, describing the information on the BitCoin market, is of the form
$$
\F_t=\F_t^W \vee \F_t^Z, \quad t \ge 0,
$$
where $\F_t^W$ and $\F_t^Z$ denote the $\sigma$-algebras generated by $W_t$ and $Z_t$ respectively up to time $t \geq 0$. Note that $\F_t^Z=\F_t^P$, for each $t \geq 0$, with $\F_t^P$ being the $\sigma$-algebra generated by $P_t$ up to time $t \geq 0$.
Since at any time $t$ the BitCoin price dynamics is affected by the sentiment index only up to time $t-\tau$, to describe the traders information on the digital market, we consider
the filtration $\widetilde \bF=\{\widetilde \F_t,\ t \geq 0\}$, defined by
$$
\widetilde \F_t = \F_t^W  \vee \F_{t-\tau}^P, \quad t \geq 0.
$$
%Note that $\F_t^Z=\F_t^P$, for each $t \geq 0$. 
We also remark that all filtrations satisfy the usual conditions of completeness and right continuity (see e.g. \citet{protter2005stochastic}).
Now, we introduce the {\em integrated information process} 
$X^\tau=\{X_{t}^\tau,\ t \ge 0\}$ associated to the sentiment index $P$, defined as follows:
\begin{equation} \label{eq:int_info}
X_t^\tau  := \left\{ 
\begin{array}{ll}
\int_0^t P_{u-\tau} \ud u =\int_{-\tau}^{0} \phi(u) \ud u + \int_0^{t-\tau} P_u \ud u= X_\tau^\tau + \int_0^{t-\tau} P_u \ud u, & \quad 0 \leq \tau \leq t,\\
\int_{-\tau}^{t-\tau} \phi(u) \ud u, & \quad 0 \leq t \leq \tau.
\end{array}
\right.
\end{equation}
Note that, for $t \in [0,\tau]$, we have $X_t^\tau=\int_{-\tau}^{t-\tau} \phi(u) \ud u$ which is deterministic.
In addition, for a finite time horizon $T>0$, let us define the corresponding variation over the interval $[t,T]$, for $t\leq T$, as $X_{t,T}^\tau:=X_T^\tau-X_t^\tau$. Obviously, $X_{T,T}^\tau=0$; moreover, for $t<T$,
\begin{equation} \label{eq:int2_info}
X_{t,T}^\tau  := \left\{ 
\begin{array}{ll}
\int_{t-\tau}^{T-\tau} P_u \ud u & \quad  \mbox{ if }\ 0 \leq \tau \leq t < T, \\
\int_{t-\tau}^{0} \phi(u) \ud u + \int_0^{T-\tau} P_u \ud u & \quad \mbox{ if } \ 0 \leq t \leq \tau < T, \\
\int_{t-\tau}^{T-\tau} \phi(u) \ud u & \quad  \mbox{ if }\ 0\leq  t < T \leq \tau. \\
\end{array}
\right.
\end{equation}
Again, note that for $T\leq \tau$, we get $X_{t,T}^\tau=\int_{t-\tau}^{T-\tau} \phi(u) \ud u$ which is deterministic.

The following lemma establishes basic statistical properties for the integrated information process as well as for its variation in case they are not fully deterministic.
%mean price of $P$ in the main market and its delayed version.
\begin{lemma}\label{th:means}
In the market model outlined above we have:
\begin{itemize} 
\item[(i)] For $t > \tau$,
\begin{align*}
\esp{X_t^\tau} & =
X_\tau^\tau +\frac{\phi(0)}{\mu_{P}}\left(e^{\mu_{P}(t-\tau)} -1\right);\\
\mathbb V{\rm ar}[X_t^\tau] & = 
\frac{2\phi^{2}(0)}{\left(\mu_{P}+\sigma_{P}^{2}\right)\left(2\mu_{P}+\sigma_{P}^{2}\right)}\left(e^{\left(2\mu_{P}+\sigma_{P}^{2}\right)(t-\tau)} -1\right)\\
& \qquad -\frac{2\phi^{2}(0)}{\mu_{P}\left(\mu_{P}+\sigma_{P}^{2}\right)}\left(e^{\mu_{P}(t-\tau)}-1\right)-\left(\frac{\phi(0)}{\mu_{P}}\left(e^{\mu_{P}(t-\tau)}-1\right)\right)^2.
\end{align*}

\item[(ii)] For $ \tau \leq t <  T $,
\begin{align*}
\esp{X_{t,T}^\tau} & = 
\frac{\phi(0)e^{\mu_{P}(t-\tau)}}{\mu_{P}}\left(e^{\mu_{P}(T-t)}-1\right);\\
\mathbb V{\rm ar}[X_{t,T}^\tau] & = 
\frac{2\phi^{2}(0)e^{\left(2\mu_{P}+\sigma_{P}^{2}\right)(t-\tau)}}{\left(\mu_{P}+\sigma_{P}^{2}\right)\left(2\mu_{P}+\sigma_{P}^{2}\right)}\left(e^{\left(2\mu_{P}+\sigma_{P}^{2}\right)(T-t)} -1\right)\\
& \qquad -\frac{2\phi^{2}(0)e^{\mu_{P}(t-\tau)}}{\mu_{P}\left(\mu_{P}+\sigma_{P}^{2}\right)}\left(e^{\mu_{P}(T-t)} -1\right)-\left(\frac{\phi(0)e^{\mu_{P}(t-\tau)}}{\mu_{P}}\left(e^{\mu_{P}(T-t)} -1\right)\right)^2.
\end{align*}

\item[(iii)] For $ t \leq \tau <  T $,
\begin{align*}
\esp{X_{t,T}^\tau} & = 
\int_{t-\tau}^0 \phi\left(u\right)\ud u + \frac{\phi(0)}{\mu_{P}}\left(e^{\mu_{P}(T-\tau)}-1\right);\\
\mathbb V{\rm ar}[X_{t,T}^\tau] & = 
\frac{2\phi^{2}(0)}{\left(\mu_{P}+\sigma_{P}^{2}\right)\left(2\mu_{P}+\sigma_{P}^{2}\right)}\left(e^{\left(2\mu_{P}+\sigma_{P}^{2}\right)(T-\tau)} -1\right)\\
& \qquad -\frac{2\phi^{2}(0)}{\mu_{P}\left(\mu_{P}+\sigma_{P}^{2}\right)}\left(e^{\mu_{P}(T-\tau)} -1\right)-\left(\frac{\phi(0)}{\mu_{P}}\left(e^{\mu_{P}(T-\tau)} -1\right)\right)^2.
\end{align*}
\end{itemize}
\end{lemma}

In \cite{hull1987pricing} similar outcomes are claimed for $\tau=0$ without providing a proof; for the sake of clarity,  we give a self-contained proof in Appendix \ref{appendix:technical}.\\
The system given by equations \eqref{eq:S} and \eqref{eq:BSdyn}
is well-defined in $\R_+$ as stated in the following theorem, which also provides its explicit solution.
\begin{theorem} \label{th:sol}
In the market model outlined above, the followings hold:
\begin{itemize}
\item[(i)] the bivariate stochastic delayed differential equation 
\begin{equation}\label{eq:Bivdyn}
\left\{
\begin{array}{ll}
\ud S_t = \mu_S P_{t-\tau} S_t \ud t+\sigma_S \sqrt{P_{t-\tau}}S_t\ud W_t, \quad  S_0=s_0 \in \R_+,  \\
\ud P_t =\mu_P P_t\ud t+\sigma_P P_t\ud Z_t, \quad P_t = \phi(t),\ t \in [-L,0], \\%\quad  P_0 >0, \\
\end{array}
\right.
\end{equation}
has a continuous, $\bF$-adapted, unique solution $(S,P)=\{(S_t,P_t),\ t \geq 0\}$ given by
\begin{align}
S_t & =s_0e^{\left(\mu_{S}-\frac{\sigma_{S}^{2}}{2}\right)\int_0^t P_{u-\tau} \ud u+\sigma_S \int_0^t \sqrt{P_{u-\tau}}\ud W_u},\quad t \ge 0,\label{eq:sol_S}\\
P_t & = \phi(0)e^{\left(\mu_{P}-\frac{\sigma_{P}^{2}}{2}\right)t+\sigma_P Z_t}, \quad t \ge 0. \label{eq:sol_P}
\end{align}
More precisely, $S$ can be computed step by step as follows: for $k=0,1,2,\ldots$ and $t \in [k\tau,(k+1)\tau]$,
\begin{equation} \label{eq:sol_S1} 
S_t  = S_{k\tau}e^{\left(\mu_{S}-\frac{\sigma_{S}^{2}}{2}\right)\int_{k\tau}^t P_{u-\tau} \ud u+\sigma_S \int_ {k\tau}^t \sqrt{P_{u-\tau}}\ud W_u}. 
\end{equation} %l'ho reinserito per avere il riferimento nella proof (31 Agosto)
In particular, %if $\phi(0) \ge 0$, then 
$P_t \ge 0$ $\P$-a.s. for all $t \geq 0$. If in addition, $\phi(0) > 0$, then $P_t > 0$ $\P$-a.s. for all $t \geq 0$.
\item [(ii)] Further, for every $t \ge 0$, the conditional distribution of $S_{t}$, given the integrated information $X_t^\tau$,
is log-Normal with mean $\log \left(s_0\right) + \left(\mu_{S}-\frac{\sigma_{S}^{2}}{2}\right)X_t^\tau$ and variance $\sigma_{S}^{2}X_t^\tau$.
\item [(iii)] Finally, for every $t \in [0,\tau]$, the random variable $\log \left( S_t \right)$ has mean $\log \left(s_0\right) + \left(\mu_{S}-\frac{\sigma_{S}^{2}}{2}\right)X_t^\tau$ and variance $\sigma_{S}^{2}X_t^\tau$; for every $t > \tau$,  $\log \left(S_t\right)$ has mean and variance respectively given by
\begin{align*}
\esp{\log \left( S_t \right)} & = \log \left(s_0\right) + \left(\mu_S-\frac{\sigma_S^2}{2}\right)  \esp{X_t^\tau};\\
\mathbb V{\rm ar}\left[{\log \left( S_t \right)}\right] & = \left(\mu_S-\frac{\sigma_S^2}{2}\right)^2\mathbb V{\rm ar}[X_t^\tau]+ \sigma_S^2 \esp{X_t^\tau},
\end{align*}
where $\esp{X_t^\tau}$ and $\mathbb V{\rm ar}[X_t^\tau]$ are both provided by Lemma \ref{th:means}, point (i).
\end{itemize}
\end{theorem}

\begin{proof}
{\bf Point (i)}. 
Clearly, $S$ and $P$, given in \eqref{eq:sol_S} and \eqref{eq:sol_P} respectively, are $\bF$-adapted processes with continuous trajectories.
Similarly to \citet[Theorem 2.1]{mao2013delay}, we provide existence and uniqueness of a strong solution to the pair of stochastic differential equations in
system \eqref{eq:Bivdyn} by using forward induction steps of length $\tau$,
without the need of checking any assumptions on the coefficients, e.g. the local Lipschitz condition and the linear growth condition.

First, note that the second equation in the system \eqref{eq:Bivdyn} does not depend on $S$, and its solution is well known for all $t \ge 0$. Clearly, equation \eqref{eq:sol_P} says that $P_t \ge 0$ $\P$-a.s. for all $t \ge 0$ and that $\phi(0) > 0$ implies that the solution $P$ remains strictly greater than $0$ over $[0,+\infty)$, i.e. $P_t > 0$, $\P$-a.s. for all $t \ge 0$.\\
Next, by the first equation in 
\eqref{eq:Bivdyn} and applying It\^o’s formula to $\log \left(S_t\right)$, we get
\begin{equation} \label{eq:logS}
\ud\log \left(S_t\right) =\left(\mu_S-\frac{\sigma_S^2}{2}\right)P_{t-\tau} \ud t + \sigma_S\sqrt{P_{t-\tau}} \ud W_t,
\end{equation}
or equivalently, in integral form
\begin{equation} \label{eq:log_int}
%\log S_t=\log S_0 
\log\left(\frac{S_{t}}{s_{0}}\right) = \left(\mu_S-\frac{\sigma_S^2}{2}\right)\int_0^t P_{u-\tau} \ud u + \sigma_S\int_0^t\sqrt{P_{u-\tau}} \ud W_u,\quad t \ge 0.
\end{equation}
For $t \in [0,\tau]$, \eqref{eq:log_int} can be written as
%$$
%\ud\log S_t =\left(\mu_S-\frac{\sigma_S^2}{2}\right)\phi \left( t-\tau \right) \ud t + \sigma_S\sqrt{\phi\left( t-\tau\right)} \ud W_t,
%$$
%or equivalently in integral form, 
%and by integrating both sides of the equation, we obtain that
\begin{equation} \label{eq:log1}
%\log S_t=\log S_0 
\log\left(\frac{S_{t}}{s_{0}}\right) = \left(\mu_S-\frac{\sigma_S^2}{2}\right)\int_0^t \phi \left(u-\tau \right) \ud u + \sigma_S\int_0^t\sqrt{\phi \left(u-\tau \right)} \ud W_u,
\end{equation}
that is, \eqref{eq:sol_S1} holds for $k = 0$.

Given that $S_t$ is now known for $t \in [0,\tau]$, we may restrict the first equation in \eqref{eq:Bivdyn} on $t \in [\tau, 2\tau]$, so that it corresponds to consider \eqref{eq:logS} for $t \in [\tau, 2\tau]$.
%Further, for $t \in (k\tau,(k+1)\tau]$, with $k=1$, 
Equivalently, in integral form,
\begin{equation} \label{eq:log}
%\log S_t=\log S_\tau 
\log\left(\frac{S_{t}}{S_{\tau}}\right) = \left(\mu_S-\frac{\sigma_S^2}{2}\right)\int_\tau^t P_{u-\tau} \ud u + \sigma_S\int_\tau^t\sqrt{P_{u-\tau}} \ud W_u.
\end{equation}
This shows that \eqref{eq:sol_S1} holds for $k = 1$.
Similar computations for $k=2,3,\ldots$, give the final result.

{\bf Point (ii)}. %To complete the proof 
Set  $Y_t:=\int_0^t\sqrt{\phi \left( t-\tau \right)} \ud W_u $, for $t \in [0,\tau]$ and $Y_t:= Y_{k\tau}+ \int_{k\tau}^t\sqrt{P_{u-\tau}} \ud W_u $, for $t \in [k\tau,(k+1)\tau]$, with $k=1,2,\ldots$. Then, by applying the outcomes in Point (i) and the decomposition
$$
\log\left(\frac{S_{t}}{s_{0}}\right) =\log\left(\frac{S_{t}}{S_{k\tau}}\right)+\sum_{j=0}^{k-1}\log\left(\frac{S_{(j+1)\tau}}{S_{j\tau}}\right), %+ \log\left(\frac{S_{\tau}}{s_0}\right)
$$
for $t\in [k\tau,(k+1)\tau]$, with $k=1,2,\ldots$, we can write
\begin{equation} \label{eq:LOG}
%\log S_t=\log S_0 
\log\left(S_{t}\right)= \log(s_0)+\left(\mu_S-\frac{\sigma_S^2}{2}\right)X_t^\tau +\sigma_S Y_t,\quad t \ge 0.
\end{equation}
To complete the proof, 
it suffices to show that, for each $t\ge 0$ the random variable $Y_t$, conditional on $X_t^\tau$, is Normally distributed with mean $0$ and variance $X_t^\tau$. This is straightforward from \eqref{eq:log1} if $t \in [0,\tau]$. Otherwise, we first observe that
since $Z_{u-\tau}$ is independent of $W_u$ for every $\tau < u \leq t$, the distribution of $Y_t$, conditional on $\{Z_{u-\tau}:\ \tau < u \leq t-\tau\}=\{P_{u-\tau}:\ \tau < u \leq t-\tau\} %=\{P_u^D:\ 0 \leq u \leq t\}
=\F_{t-\tau}^P$, is Normal with mean $0$ and variance $\sigma_S^2X_t^\tau$. \\
Now, for each $t > \tau$, the moment-generating function of $Y_t$, conditioned on the history of the process $P$ up to time $t-\tau$, is given by
\begin{align*}
\esp{e^{aY_t}\Big{|}\F_{t-\tau}^{P}}
& = e^{\int_0^t \frac{a^2}{2} P_{u-\tau} \ud u} =  e^{\frac{a^2}{2}  \int_0^t  P_{u-\tau} \ud u}\\
& = e^{\frac{a^2}{2}\left(\sqrt{X_t^\tau}\right)^2}, \quad  a \in \R,
\end{align*}
that only depends on its integrated information $X_t^\tau$ up to time $t$, that is,  
$$
\esp{e^{aY_t}\Big{|}\F_{t-\tau}^{P}}=\esp{e^{aY_t}\Big{|}X_t^\tau}, \quad t > \tau.
$$

{\bf Point (iii)}. The proof is trivial for $t \in [0,\tau]$. If $t >\tau$, \eqref{eq:LOG} and Lemma \ref{th:means} together with the null-expectation property of the It\^o integral, %for each $t \in (0,T]$, 
give
\begin{align*}
\esp{\log \left(S_t\right)} & = \log \left(s_0\right) + \left(\mu_S-\frac{\sigma_S^2}{2}\right) \esp{X_t^\tau}\\
& = \log \left(s_0\right) + \left(\mu_S-\frac{\sigma_S^2}{2}\right)\left(X_\tau^\tau  +\frac{\phi(0)}{\mu_{P}}\left(e^{\mu_{P}(t-\tau)} -1\right)\right).
\end{align*}
Now, we compute the variance of $\log \left(S_t\right)$. %for each $t \in [0,T]$, 
Since for each $t > \tau$ the random variable $Y_t $ has mean $0$ conditional on $\F_{t-\tau}^P$, we have
\begin{align*}
\mathbb V{\rm ar}\left[\log \left(S_t\right)\right] & = \esp{\log^2\left(S_t\right)} - \left(\esp{\log\left(S_t\right)}\right)^2\\
& = \left(\mu_S-\frac{\sigma_S^2}{2}\right)^2\esp{(X_t^\tau)^2}  + 2\left(\mu_S-\frac{\sigma_S^2}{2}\right)\esp{X_t^\tau \esp{Y_t\Big{|}\F_{t-\tau}^P}}\\
& \qquad \qquad + \sigma_S^2 \esp{X_t^\tau} - \left(\mu_S-\frac{\sigma_S^2}{2}\right)^2 \esp{X_t^\tau}^2\\
& = \left(\mu_S-\frac{\sigma_S^2}{2}\right)^2\mathbb V{\rm ar}[X_t^\tau]+ \sigma_S^2 \esp{X_t^\tau}.
\end{align*}

Thus, the proof is complete.
\end{proof}

\subsection{Existence of a risk-neutral probability measure}
Let us fix a finite time horizon $T>0$ 
and assume the existence of a riskless asset, say the money market account, whose value process $B=\{B_t,\ t \in [0,T]\}$ %at time $t \in [0,T]$ 
is given by
$$
B_t=e^{\int_0^t r(s)\ud s},\quad t \in [0,T],
$$
where $r:[0,T] \to \R$ is a bounded, deterministic function %$r=\{r_t,\ t \in [0,T]\}$ is a bounded, nonnegative, $\bF$-progressively measurable process 
representing the instantaneous risk-free interest rate.
To exclude arbitrage opportunities, we %assume that the set of all equivalent martingale measures for the BitCoin price process $S$ is non-empty. Note that, it contains more than a single element, since $P$ does not represent the price of any tradeable asset, and therefore the underlying market model is incomplete.
need to check that the set of all equivalent martingale measures for the BitCoin price process $S$ is non-empty. More precisely, it contains more than a single element, since $P$ does not represent the price of any tradeable asset, and therefore the underlying market model is incomplete.
\begin{lemma}\label{lem:measure}
Let $\phi(t) > 0$, for each $t \in [-L,0]$, in \eqref{eq:BSdyn}. Then, every equivalent %(local) 
martingale measure $\Q$ for $S$ defined on $(\Omega,\F_T)$ has the following density 
%$L=\{L_t,\ t \in [0,T]\}$, given by
\begin{equation}\label{def:Q}
%\begin{split}
%L_t & :=
\frac{\ud \Q}{\ud \P}\bigg{|}_{\F_T}=:L_T^\Q, \quad \P-\mbox{a.s.},
%\end{split}
\end{equation} 
where $L_T^\Q$ is the terminal value of the $(\bF,\P)$-martingale $L^\Q=\{L_t^\Q,\ t \in [0,T]\}$ given by
\begin{equation} \label{eq:L}
L_t^\Q :=\doleans{-\int_0^\cdot \frac{\mu_S P_{s-\tau}-r(s)}{\sigma_S\sqrt{P_{s-\tau}}} \ud W_s - \int_0^\cdot\gamma_s\ud Z_s}_t, \quad t \in [0,T], 
%- \frac{1}{2}\int_0^t\left(\frac{\mu_S P_{s-\tau}-r(s)}{\sigma_S\sqrt{P_{s-\tau}}}\right)^2\ud s -\frac{1}{2}\int_0^t\gamma_s^2\ud s}_t, \quad t \in [0,T], 
\end{equation}
for a suitable $\bF$-progressively measurable process $\gamma=\{\gamma_t,\ t \in [0,T]\}$.
\end{lemma}
The proof is postponed to Appendix \ref{appendix:technical}. Here $\E (Y)$ denotes the Doleans-Dade exponential of an $(\bF, \P)$-semimartingale $Y$.

In the rest of the paper, suppose that $\phi(t) > 0$, for each $t \in [-L,0]$, in \eqref{eq:BSdyn}.
Then, Lemma \ref{lem:measure} ensures that the space of equivalent martingale measures for $S$ is described by 
\eqref{eq:L}. More precisely, it is
parameterized by the process $\gamma$
which governs the change of drift of the  $(\bF,\P)$-Brownian motion $Z$. Note that
the sentiment factor dynamics under $\Q$ in the BitCoin market is given by
$$
\ud P_t  = (\mu_P - \sigma_P \gamma_t)P_t\ud t + \sigma_P P_t \ud Z_t^\Q, \quad P_t = \phi(t),\ t \in [-L,0].
$$
The process $\gamma$ %is termed the market price of volatility risk or volatility risk premium.
can be interpreted as the risk perception associated to the future direction or future possible movements of the BitCoin market.
One simple example of a candidate equivalent martingale measure is the so-called {\em minimal martingale measure} (see e.g. \citet{follmer1991hedging}, \citet{follmer2010minimal}), 
denoted by $\widehat \P$, whose density process $L = \{L_t,\ t \in  [0,T]\}$, is given by
\begin{equation} \label{eq:Lp}
L_t := e^{-\int_0^t \frac{\mu_S P_{s-\tau}-r(s)}{\sigma_S\sqrt{P_{s-\tau}}} \ud W_s - \frac{1}{2}\int_0^t\left(\frac{\mu_S P_{s-\tau}-r(s)}{\sigma_S\sqrt{P_{s-\tau}}}\right)^2\ud s},\quad t \in [0,T]. 
\end{equation}

This is the probability measure which corresponds to the choice $\gamma \equiv 0$ in \eqref{eq:L}. %once we require that $L$ is square-integrable. 
Intuitively, under the minimal martingale measure, say $\widehat \P$, 
the drift of the Brownian motion driving the BitCoin price process $S$ is modified to make $S$ an $(\bF,\widehat \P)$-martingale, while the drift of the Brownian motion which is strongly orthogonal to $S$ is not affected by the change measure from $\P$ to $\widehat \P$. More precisely, under the change of measure from $\P$ to $\widehat \P$, we have two independent $(\bF,\widehat \P)$-Brownian motions $\widehat W=\{\widehat W_t,\ t \in [0,T]\}$ and $\widehat Z=\{\widehat Z_t,\ t \in [0,T]\}$ defined respectively by
\begin{align}
\widehat W_t & := W_t + \int_0^t\frac{\mu_S P_{s-\tau}-r(s)}{\sigma_S\sqrt{P_{s-\tau}}} \ud s,\quad t \in [0,T],\\
\widehat Z_t & := Z_t, \quad t \in [0,T]. \label{def:hat_Z}
\end{align}
%Under the probability measure $\widehat \P$,
Denote by  $\widetilde S_t=\{\widetilde S_t,\ t \in [0,T]\}$ the discounted BitCoin price process 
defined as $\ds \widetilde S_t:=\frac{S_t}{B_t}$, for each $t \in [0,T]$.
Then, on the probability space $(\Omega,\F,\widehat \P)$, the pair $(\widetilde S,P)$ satisfies the following system of stochastic delayed differential equations:
\begin{equation}\label{eq:Bivdyn-MMM}
\left\{
\begin{array}{ll}
\ud \widetilde S_t = \sigma_S \sqrt{P_{t-\tau}}\widetilde S_t\ud \widehat W_t, \quad  \widetilde S_0=s_0 \in \R_+,  \\
\ud P_t =\mu_P P_t\ud t+\sigma_P P_t\ud Z_t, \quad P_t = \phi(t),\ t \in [-L,0], \\%\quad  P_0 >0, \\
\end{array}
\right.
\end{equation}

By Theorem \ref{th:sol}, point (i), the explicit expression of the solution to \eqref{eq:Bivdyn-MMM}, which provides the 
discounted BitCoin price $\widetilde S_t$, at any time $t \in [0,T]$, is given by
\begin{equation}\label{def:tilde_S}
\widetilde S_t=s_0e^{\sigma_S\int_0^t\sqrt{P_{u-\tau}}\ud \widehat W_u - \frac{\sigma_S^2}{2}\int_0^t P_{u-\tau} \ud u},\quad t \in [0,T],
\end{equation}
with the representation of the sentiment factor $P$ still provided by \eqref{eq:sol_P}.

Under our assumptions, the dynamics of the (non-discounted) BitCoin price under the minimal martingale measure is given by 
\begin{equation}\label{eq:Bivdyn-Q}
\left\{
\begin{array}{ll}
\ud S_t = r(t) \ud t + \sigma_S \sqrt{P_{t-\tau}} S_t\ud \widehat W_t, \quad  S_0=s_0 \in \R_+,  \\
\ud P_t =\mu_P P_t\ud t+\sigma_P P_t\ud Z_t, \quad P_t = \phi(t),\ t \in [-L,0], \\%\quad  P_0 >0, \\
\end{array}
\right.
\end{equation}
where $r(t)$ is the risk-free interest rate at time $t$.
The above dynamics was assumed in \citet{hull1987pricing} to describe price changes for a stock and its instantaneous variance (for which $\sigma_S=1$ and $\tau=0$ by definition). However, the authors assumed from the very beginning a risk-neutral framework without defining the dynamics under the physical measure and with no proof of the existence of any equivalent martingale measure. 

%a Geometric Brownian motion with the same input coefficients, i.e. 
%$$
%\ud P_t  = \mu_P P_t\ud t + \sigma_P P_t \ud \widehat Z_t, \quad P_t = \phi(t),\ t \in [-L,0].
%$$

%\begin{ass} \label{ass:minimal}
%Set $\gamma \equiv 0$ in \eqref{eq:L} and the density process $L$ is square-integrable.
%\end{ass}
%\fbox{AGGIUNGERE:} 
%\begin{itemize}
%\item FRASE CONCLUSIVA
%\item FRASE DI RACCORDO CON LA SECTION
% SUCCESSIVA
% \item ACCENNO ALLA SECTION SULL'OPTION PRICING
%\end{itemize}

In Section \ref{sec:option-pricing}, we derive option pricing formulas via the risk-neutral evaluation procedure based on the minimal martingale measure above defined. As usual, pricing formulas depend on model parameters which have to be estimated on market data. A common approach, when a closed formula for option is available, is the so called calibration of parameters; their value is obtained in order to minimize a proper distance between model an market prices for options. However, this method is particularly of interest when there is a standardized and liquid market for options. 
Of course, this is not the case for the BitCoin so we will fit the model directly to a time series of BitCoin prices with a more classical statistical procedure based on the approximation of the probability density function of a discrete sample for model described by equations \eqref{eq:S} and \eqref{eq:BSdyn}. To this end we need to know the dynamics of the BitCoin price under the physical measure and derive statistical properties for a discrete sample of the process $P$ given in \eqref{eq:BSdyn}.

\section{Statistical properties of discretely observed quantities and parameter estimation} \label{sec:loglike}

In this section, we introduce basic statistical properties for a sample of discretely observed prices and suggest a possible closed form approximation for the joint probability density of the discrete sample.

Let us fix a discrete observation step $\Delta$ and consider the discrete time process $\{S_i,\ i \in \bN\}$, where $S_i:=S_{i\Delta}$. Define the corresponding logarithmic returns process $\{R_i,\ i\in \bN \}$ as
\begin{equation}
\label{eq:Ri} 
R_i=\log (S_{i})-\log (S_{i-1}).
\end{equation}
By \eqref{eq:LOG}, we get
\begin{equation} \label{def:R}
R_i= \left(\mu_S-\frac{\sigma_S^2}{2}\right)\int_{(i-1)\Delta}^{i\Delta} P_{u-\tau} du + \sigma_S \int_{(i-1)\Delta}^{i\Delta} \sqrt{P_{u-\tau}} \ud W_u, \quad i \in \bN.
\end{equation}
Setting $Y_t:=\int_0^t \sqrt{P_{u-\tau}} \ud W_u$, with $t \in [0,T]$,
as in the proof of Theorem \ref{th:sol}, we define  
\begin{equation}
\label{eq:Yi}
M_i:=Y_{i\Delta}-Y_{(i-1)\Delta}=\int_{(i-1)\Delta}^{i\Delta} \sqrt{P_{u-\tau}} \ud W_u,\,\,\, \forall i \in \bN,
\end{equation}
so that, \eqref{def:R} can be written as
\begin{equation}
R_i= \left(\mu_S-\frac{\sigma_S^2}{2}\right)A_i^\tau + \sigma_S M_i, \quad i \in \bN,
\end{equation}
where $A_i^\tau:=X_{(i-1)\Delta,\; i\Delta}^\tau$, with $X_{t,T}^\tau$ being the variation of the integrated information process introduced in \eqref{eq:int_info}; since $\tau$ is fixed we omit hereafter the dependence on it and, without loss of generality we assume $\tau < \Delta$ so that $A_1=X_\tau^\tau + \int_0^{\Delta-\tau} P_u \ud u$. Note that if $j\Delta \leq \tau<(j+1)\Delta $ the quantities $A_1,\ldots,A_j$ are deterministic and the outcomes in what follows still hold if $A_1$ is replaced by the first non deterministic value $A_{j+1}$. 
 
Let us consider a finite time horizon $T=n\Delta$; under model assumptions the conditional probability distribution of the vector $\mathbf{M}=\left(M_1,M_2,\dots,M_n\right)$, given the vector $\mathbf{A}=\left(A_1,A_2,\dots,A_n\right)$, is a multi-variate normal with covariance matrix $Diag(A_1,A_2,\dots,A_n)$. Hence, the vector of discretely observed logarithmic returns $\mathbf{R}=\left(R_1,R_2,\dots,R_n\right)$, conditionally on $\mathbf{A}$, is jointly normal with covariance matrix  $\Sigma=\sigma_S^2 Diag(A_1,A_2,\dots,A_n)$.
%
%Given the time istants $t_i=i\Delta$, for $i=0,1,2,...,n$, assume also we are given a discrete sample for $\left(\mathbf{R},\mathbf{A}\right)$, namely $\{(R_i,A_i))\}_i$ for $i=1,2,...,n$  and that we observe $P_0$. 
The application of Bayes's rule allows to write the unconditional joint probability distribution of $\left(\mathbf{R},\mathbf{A}\right)$, i.e. the density function $f_{\left(\mathbf{R},\mathbf{A}\right)}:\R^n \times \R_+^n \longrightarrow \R$ as
\begin{equation}\label{aprdens}
f_{\left(\mathbf{R},\mathbf{A}\right)}(\mathbf{r},\mathbf{a})=f_{A_1}(a_1)\prod_{i=2}^nf_{\left(A_i|A_{i-1}\right)}(a_i) \prod_{i=1}^n \frac{1}{\sqrt{2\pi \sigma_S^2 a_i}} e^{-\frac{1}{2}\frac{\left(r_i-\left(\mu_S-\frac{\sigma_S^2}{2}\right)a_i\right)^2}{\sigma_S^2 a_i}}.
\end{equation}
with $\mathbf{r}=\left(r_1,r_2,\dots,r_n\right) \in \R^n$ and ,
$\mathbf{a}=\left(a_1,a_2,\dots,a_n\right) \in \R_+^n$.
%and $r_i \in \R$, $a_i \in \R_+$, for every $i \in \bN$.

The probability distribution functions $f_{A_1}(.)$ and  $f_{A_i|A_{i-1}}(.) \mbox{ for } i=2,3,\dots,n$ are not available in closed form; though, several approximations exist among which those introduced in \citet{levy1992pricing} and \citet{ MilPosner}. Of course any approximation available for such densities can be applied in order to find a closed formula approximating the joint density $f_{\left(\mathbf{R},\mathbf{A}\right)}\left(\mathbf{r},\mathbf{a}\right)$; in what follows we adopt the one suggested in \citet{levy1992pricing}, see Appendix \ref{appendix:levy} for further details.  Note that the inverse gamma approach suggested in \citet{MilPosner} holds in the limit when $T$ tends to infinity, a condition which is not at all consistent with the applications we have in mind; further discussion on the approximating distribution to select is beyond the scope of our paper.

\subsection{The approximated likelihood}
One of the pillar in statistical inference is the maximum likelihood (in short ML) estimation approach where model parameters are estimated so as to maximize the probability of the the realized sample to be extracted randomly;  the likelihood function  shares the same mathematical expression of the probability density function but it is computed "ex-post" when a realization of involved random variables is available and assuming the underlying model parameters to be unknown. 
It is well known that ML estimates are consistent and asymptotically normal and they achieve efficiency, i.e. they have the lowest variance among estimators sharing the same asymptotic properties (see \citet{Davison:statmod}). 

By applying the approximation of \citet{levy1992pricing}, we prove the following Lemma. 
\begin{lemma}\label{lemma:distr}
Let $\phi(t) > 0$, for each $t \in [-L,0]$, in \eqref{eq:BSdyn} and $\tau<\Delta$. Then, in the market model outlined in Section \ref{sec:model}, we have
\begin{itemize}
\item[(i)] the distribution of $A_1-X_\tau^\tau$ is approximated by a log-normal with mean $\alpha_1$ and variance $\nu_1^2$ given by 
\begin{gather}
\alpha_1 = \log \phi(0) + 2\log \frac{e^{\mu_P(\Delta-\tau)}-1}{\mu_P} -\frac{1}{2}\log\left( \frac{2}{\mu_P+\sigma_P^2} \left[ \frac{e^{(2\mu_P+\sigma_P^2)(\Delta-\tau)}-1}{2\mu_P+\sigma_P^2}-\frac{e^{\mu_P(\Delta-\tau)}-1}{\mu_P}   \right] \right) \\
\nu_1^2=\log \left( \frac{2}{\mu_P+\sigma_P^2} \left[ \frac{e^{(2\mu_P+\sigma_P^2)(\Delta-\tau)}-1}{2\mu_P+\sigma_P^2}-\frac{e^{\mu_P(\Delta-\tau)}-1}{\mu_P} \right] \right)-2\log \left(\frac{e^{\mu_P(\Delta-\tau)}-1}{\mu_P} \right)
\end{gather}
\item[(ii)] the distribution of $A_i$ given $A_{i-1}$ (shortly $A_i|A_{i-1}$), for $i=1,\dots,n$, is approximated by a log-normal with means $\alpha_i$ and variances $\nu_i^2$ given by 
\begin{gather}
\alpha_i=\log \left(A_{i-1}\right)+ \left(\mu_P-\frac{\sigma_P^2}{2}\right)\Delta, \quad \mbox{  for } i=1,\dots,n,\\
\nu_i^2=\sigma_P^2\Delta, \quad \mbox{  for } i=1,\dots,n.
\end{gather}
\end{itemize}

\end{lemma}

The proof is postponed to Appendix \ref{appendix:technical}.
Now, we are in the position to state the following theorem.
\begin{theorem}\label{likelihood}
Under the same assumptions of Lemma \ref{lemma:distr}, given the realized sample $\left(\bar{\mathbf{r}},\bar{\mathbf{a}}\right)$, the log-likelihood function $\log{\L_\mathbf{R,A}(\mu_P,\mu_S,\sigma_P,\sigma_S)}:\R^2 \times \R_+^2 \longrightarrow \R$ can be approximated by
\begin{equation}\label{eq:likli}
\begin{split}
\log{\L_\mathbf{R,A}(\mu_P,\mu_S,\sigma_P,\sigma_S)} & =\sum_{i=1}^n\left[\log\left(\frac{1}{\sqrt{2\pi \sigma_S^2 a_i}}\right)-\frac{1}{2}\frac{\left(r_i-\left(\mu_S-\frac{\sigma_S^2}{2}\right)a_i\right)^2}{\sigma_S^2 a_i}\right]\\
& +\sum_{i=1}^n\left[\log\left( \frac{1}{a_i\nu_i\sqrt{2\pi}}\right) -\frac{\left(\log (a_i)-\alpha_i \right)^2}{2\nu_i^2}\right],
\end{split}
\end{equation}
where upper case letters are used for random variables and lowercase for the corresponding realizations.
\end{theorem}

\begin{proof}
First, recall that the likelihood function of a parameter corresponds to the probability density function where random variables are replaced by they realizations and parameters are unknown. Then,
by simply applying the logarithmic function to \eqref{aprdens} we get
\begin{equation}
\label{eq:lik1}
\begin{split}
\log{f_{(\mathbf{R,A})}(\mathbf{r},\mathbf{a})}& =\log{\left(\prod_{i=1}^n \left[ \frac{1}{\sqrt{2\pi \sigma_S^2 a_i}}e^{-\frac{1}{2}\frac{\left(r_i-\left(\mu_S-\frac{\sigma_S^2}{2}\right)a_i\right)^2}{\sigma_S^2 a_i}} \right] f_{A_1}(a_1)\prod_{i=2}^nf_{A_i|A_{i-1}}(a_i)\right)}\\
	& =\sum_{i=1}^n\left[\log\left(\frac{1}{\sqrt{2\pi \sigma_S^2 a_i}}\right)-\frac{1}{2}\frac{\left(r_i-\left(\mu_S-\frac{\sigma_S^2}{2}\right)a_i\right)^2}{\sigma_S^2 a_i}\right]+\log \left(f_{A_1}(a_1)\right) \\
	& \qquad \qquad +\sum_{i=2}^n\log \left(f_{A_i|A_{i-1}}(a_i)\right).
\end{split}
\end{equation}
Replacing the unknown densities in \eqref{eq:lik1} according to Lemma \ref{lemma:distr} gives the desired result.

\end{proof}

Maximum likelihood estimates for the model can be obtained by maximizing the log-likelihood approximation in \eqref{eq:likli} i.e.
\begin{equation}
\label{eq:maxLik}
(\widehat{\mu}_P,\widehat{\mu}_S,\widehat{\sigma}_P,\widehat{\sigma}_S)=\arg\max_{\substack{\mu_P,\mu_S \\ \sigma_P,\sigma_S}}\log{\L_\mathbf{R,A}(\mu_P,\mu_S,\sigma_P,\sigma_S)}.
\end{equation}
In this case the methodology is referred to as Quasi-Maximum likelihood since the exact expression of the likelihood is not available; under suitable conditions, quasi-maximum likelihood estimates are asymptotically equivalent to the maximum likelihood estimates, see e.g. \citet{White, Gourieroux}. We also performed a simulation study to assess finite sample behavior of the estimates. 

It is worth to stress that the above estimation method does not assume the process $P$ to be observed, as far as $X_\tau^\tau$ and $A_i$, $i \geq 1$ are observed (note that $A_i$ is the cumulative of $P$ along the time interval $[(i-1)\Delta-\tau,i\Delta-\tau]$).

\subsection{Finite sample behavior of QML estimates}\label{sec:liksimulstudy}
In order to check the goodness of the log-likelihood approximation introduced in Theorem \ref{likelihood}, we apply the proposed estimation method to simulated data and assume, for the sake of simplicity, $\tau=0$. We simulate $m$ samples of length $n$ for the processes in \eqref{eq:S} and \eqref{eq:BSdyn} assuming a constant finer observation step $\delta$; we extract corresponding samples for $\mathbf{R,A}$ at a lower frequency, with observation step $\Delta=r\delta$.  In the numerical exercise we choose $n=730$, $m=1000$, $\delta=\frac{1}{365}$ (daily observations), $\Delta=7\delta$ (weekly observations); parameters values are set as $\mu_P=2$, $\sigma_P=0.5$ $\mu_S=0.05$, $\sigma_S=0.3$. We end with $1000$ samples of $104$ observations for $\left(\mathbf{R,A}\right)$; for each sample we estimate the parameters by means of the quasi-maximum likelihood as suggested in previous subsection.
The results are summed up in Table \ref{tab:fitpar}.

\begin{table}[htbp]
\caption{Parameter fit with simulated data of QML method.}	
\centering
\captionsetup{justification=centering}
\label{tab:fitpar}
\begin{tabular}{lrrrrrr}
	\toprule
    Variable    & Theor. value & Fitted value  &  Std. error &  t-value & P($>\lvert t \rvert$) & RMSE\\ 
	  \midrule
	$\mu_P$		& 2.0000 &	1.9759	& 0.3675 &	-0.0655	& 0.9478 & 11.6475 \\
	$\sigma_P$	& 0.5000 &	0.4089	& 0.0302 &	-3.0170	& 0.0026 &	3.0358 \\
	$\mu_S$		& 0.0500 &	0.0497	& 0.0112 &	-0.0297	& 0.9763 &	0.3544 \\
	$\sigma_S$	& 0.3000 &	0.2978	& 0.0210 &	-0.1032	& 0.9178 &	0.6687\\
	\bottomrule
	\end{tabular} 
\end{table}

We also performed a $t$-test in order to check for estimation bias. The fitted values of $\mu_P,\mu_S,\sigma_S$ are close in mean to their theoretical value and with a reasonable standard deviation; the $p$-values  of the t-test confirm that estimated are not biased. Different conclusions are in order as for parameter $\sigma_P$ which estimations is by no doubt biased. In Figure \ref{fig:hist_par} we plot the histograms of the estimated as well as the fitted normal distribution and the expected mean of the asymptotic distribution. Pictures confirm the biasedness of the estimator for $\sigma_P$ but all other estimates perform well and outcomes may become better by increasing the sample length. 
The simulation exercise have been repeated by letting the parameters values, the number and the sample length vary obtaining analogous qualitative results. %\ref{sec:sim_study_par} and \ref{sec:sim_study_par2}. 

\begin{figure}[htbp]
\includegraphics[scale=0.38]{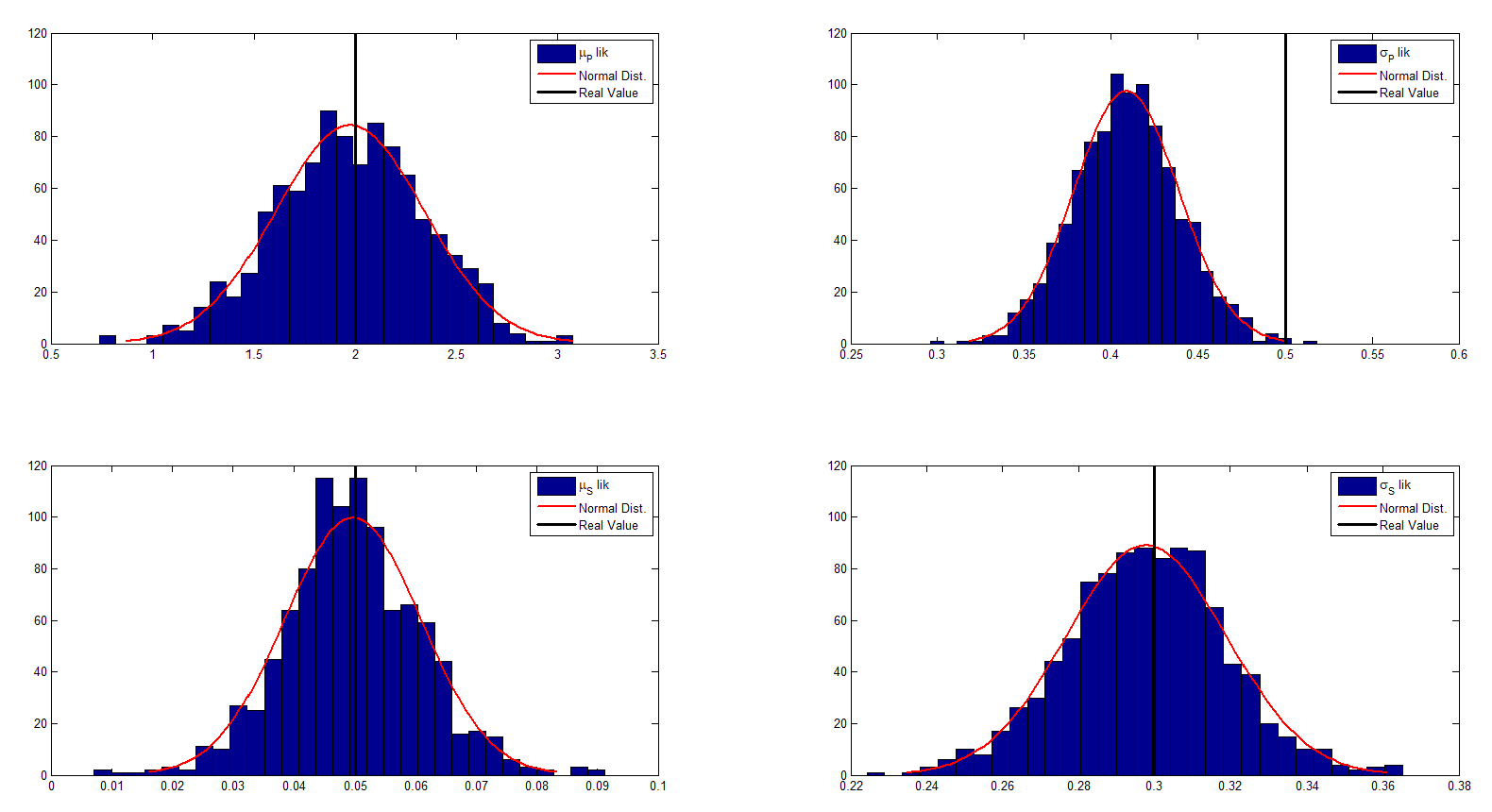} 
\caption{Histogram of parameter fit with simulated data of QML method.}\label{fig:hist_par}
\end{figure}

In order to disentangle the contribution of the Levy approximation \cite{levy1992pricing} to the estimation bias, we suggest to apply the method of moments to estimate $\mu_P$ and $\sigma_P$ considering the whole sample for $P$ generated at the finer observation step $\delta$ to compute the sample mean and sample variance of the sentiment realizations.
If we then we plug the estimated values in the likelihood \eqref{eq:likli} in order to estimate $\lbrace \mu_S, \sigma_S \rbrace$ these two estimates remain unchanged; in fact the likelihood may be maximized separately with respect to $\mu_P,\sigma_P$ and $\mu_S,\sigma_S$ since each of the two addend in the likelihood expression depends on just one of this pairs.
The results of this alternative estimation method are reported in Table \ref{tab:fitpar_momPi}.            

\begin{table}[htbp]
\caption{Parameter fit with simulated data of Moments method.}	\centering
\captionsetup{justification=centering}
\label{tab:fitpar_momPi}
\begin{tabular}{lrrrrrr}
	\toprule
    Variable    & Theor. value & Fitted value  &  Std. error &  t-value & P($>\lvert t \rvert$) & RMSE\\ 
	 \midrule
	$\mu_P$		& 2.0000 &	2.0104 & 0.3638	&  0.0286 &	0.9772 & 11.5034 \\
	$\sigma_P$	& 0.5000 &	0.4995 & 0.0135	& -0.0385 &	0.9693 &  0.4282 \\
	%$\mu_S$		& 0.0500 &	0.0497 & 0.0112	& -0.0297 &	0.9763 &  0.3544 \\
	%$\sigma_S$	& 0.3000 &	0.2978 & 0.0210	& -0.1032 &	0.9178 &  0.6686 \\
	\bottomrule
	\end{tabular} 
\end{table} 

\begin{figure}[htbp]
\includegraphics[scale=0.38]{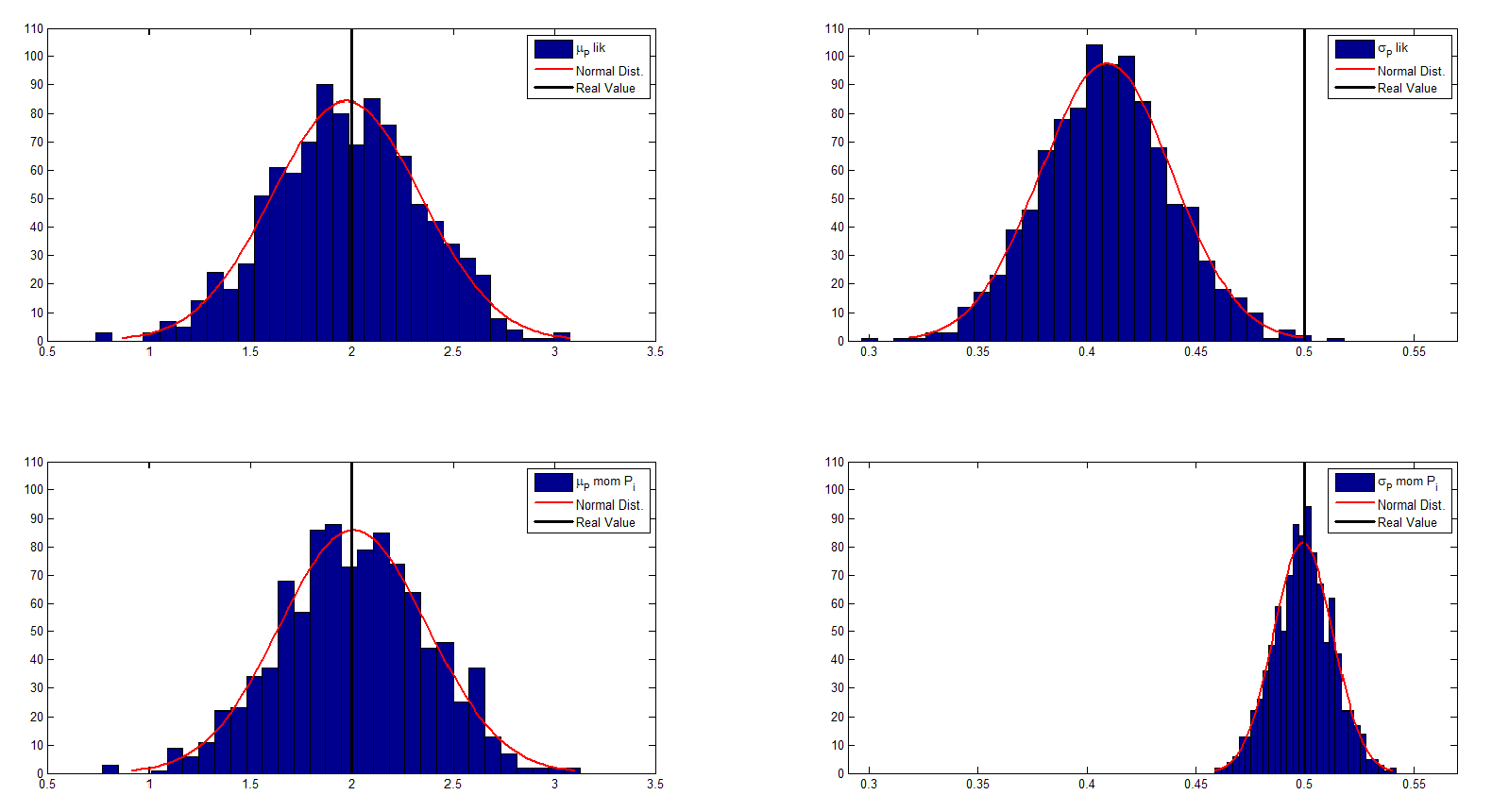} 
\caption{Histogram of parameter $\lbrace \mu_P, \sigma_P \rbrace$ fit with simulated data of QML method and Moments method at finer step $\delta$.}\label{fig:hist_par_QMLmomPi}
\end{figure}

To visualize the bias of $\lbrace \sigma_P \rbrace$ we plot in Figure \ref{fig:hist_par_QMLmomPi} the histogram of parameter $\lbrace \mu_P, \sigma_P \rbrace$ fit with simulated data using the two methods. As we can see using the two step procedure we obtain better estimates of $\lbrace \sigma_P \rbrace$ both in terms of expected value and standard deviation.
It is evident from Table \ref{tab:fitpar_momPi} and Figure \ref{fig:hist_par_QMLmomPi} the that the estimation of $\sigma_P$ is not biased in this case hence the estimation bias may essentially be attributed to the aggregation of the sentiment over time intervals and to the corresponding approximating distribution. Hence, whether the sentiment factor is observed at a finer step than the price, the above separate estimation is more reliable.

\subsection{Estimation of the delay parameter}
The delay parameter $\tau$ directly affects the definition of the discrete process $A_i$. Hence, in order to proceed with its estimation we need to observe the process $P$ at a finer observation step $\delta$ with respect to the log-returns. 
In what follows we set $\Delta=\delta r$ and we adopt a two step estimation procedure known as Profile Likelihood in order to estimate the delay. We briefly describe the Profile Likelihood approach to estimation and its application in our specific case; interested readers are referred to \citet{Davison:statmod, Pawitan:allLik} for details on the profile likelihood. 
The basic idea of this approach is to split the parameter vector which has to be estimated, say $\theta$, in two sub-vectors, one representing the parameter of interest and the other the so called nuisance parameter i.e. $\theta=(\gamma,\lambda)$;
to estimate $\gamma$ and $\lambda$ jointly we should maximize at once the likelihood i.e.
\begin{equation}
\max_{\substack{\gamma,\lambda}}\log{\L\left(\left(\gamma,\lambda\right)\right)}.
\end{equation}
When this is not feasible and provided the likelihood computed with respect to the nuisance parameter vector $\lambda$ is available and it is easy to maximize we can apply a two step procedure by maximizing, $\forall \gamma$ in its parametric space,
\begin{equation}
\label{eq:proflik}
\mathcal{L}_p(\gamma)=\max_{\substack{\lambda}}\mathcal{L}(\gamma,\lambda)=\mathcal{L}(\gamma,\widehat{\lambda}_\tau),\,\,\, 
\end{equation}
where $\widehat{\lambda}_\gamma$ is the maximum likelihood estimate of $\lambda$ for a fixed $\gamma$, then the best estimate for $\tau$ is
\begin{equation}
\widehat{\gamma}=\arg \max_{\substack{\gamma}}\log{\mathcal{L}_p(\gamma)}.
\end{equation}

%
%\begin{equation}
%\max_{\substack{\tau,\lambda}}\log{\L_\mathbf{R,A}\left(\left(\tau,\lambda\right),\mathbf{r},\mathbf{a}\right)}.
%\end{equation}
%However, this is not feasible for our model, as already noticed.  Nevertheless, the likelihood for the nuisance parameter $\lambda$ can be approximated in closed form. 
%Hence, we apply the profile likelihood approach which performs the maximization in two steps (\citet{Davison:statmod, Pawitan:allLik}).
%Define

Classical confidence intervals cannot be defined in this setting; indeed, it is possible to obtain a confidence region for $\tau$ using the likelihood ratio statistics (see \citet{Davison:statmod}),  defined as
\begin{equation}
W_p(\gamma_0)=2\left\lbrace \mathcal{L}(\widehat{\gamma},\widehat{\lambda}) - \mathcal{L}(\gamma_0,\widehat{\lambda}_{\gamma_0}) \right\rbrace,
\end{equation}
where
\begin{equation}
W_p(\gamma_0) \xrightarrow{ \;D\; } \chi_p^2
\end{equation}
and $\gamma_0$ is an assigned value for $\gamma$.
These results imply that the confidence region for $\gamma$ is the set
\begin{equation}
\label{eq:confreg}
\left\lbrace \gamma : \mathcal{L}_p(\gamma)\geq \mathcal{L}_p(\widehat{\gamma})-\frac{1}{2} c_p(1-2\alpha) \right\rbrace ,
\end{equation}
with $c_p(\alpha)$ is the $\alpha$ quantile of the $\chi_p^2$ distribution.

In our exercise we split $\theta:=(\mu_P,\mu_S,\sigma_P,\sigma_S,\tau)$ in $\theta=(\tau,\lambda)$ where $\tau$ is the parameter on which we are focusing  and $\lambda=(\mu_P,\mu_S,\sigma_P,\sigma_S)$ is the nuisance parameter vector.  The Profile Likelihood approach is feasible in our case since a closed approximating expression for the likelihood with respect to the nuisance parameter is indeed available.
%
%At first,  we assume $\tau=\tau^0$ is fixed;  $\lambda$ cam be estimated by maximizing the function in Lemma \ref{lemma} with respect to $\lambda$, i.e.
%
%\begin{equation}
%\widehat{\lambda}_{\tau_0}=\arg \max_{\substack{\lambda}}\log{f_\mathbf{R,A}(\mathbf{r},\mathbf{a})}.
%\end{equation}
%
%The model with parameters $\theta_0=(\tau_0,\lambda)$  is nested within the more general model with parameters $\theta=(\tau,\lambda)$, and of course $\mathcal{L}(\widehat{\tau},\widehat{\lambda})\geq \mathcal{L}(\tau_0,\widehat{\lambda}_{\tau_0})$.
The parametric space for $\gamma:=\tau$ is the interval $[0,L]$ in this case but, for practical purposes, $\tau$ is chosen on a grid i.e.  $\tau\in\lbrace \tau_0,\tau_1,\tau_2,\ldots,\tau_k\rbrace$; the maximization of the likelihood  $\log{\L_\mathbf{R,A}(\mathbf{r},\mathbf{a})}$ is then performed with respect to $\lambda$ for each value $\tau_j$ in the grid, obtaining $\mathcal{L}_p(\tau_j)$ for $j=0,1,\dots,k$. An estimate for $\tau$ is then obtained as $\widehat{\tau}=\arg\max_j \mathcal{L}_p(\tau_j)$.
Finally we get  $\widehat{\theta}=\left(\widehat{\tau},\widehat{\lambda}_{\widehat{\tau}} \right)$.
Of course the estimation error decreases with the mesh of the grid so that it sufficiently spans the parametric set for $\tau$.

\section{Risk neutral evaluation of European-type contingent claims} \label{sec:option-pricing}

Let $H=\varphi(S_T)$ be an $\widetilde \F_T$-measurable random variable representing the payoff a European-type contingent claim with date of maturity $T$, which can be traded on the underlying market. Here $\varphi : \R \to \R$ is a
a Borel-measurable function  such that $H$ is integrable under $\widehat \P$. The function $\varphi$ is usually referred to as the {\em contract function}.
The following result provides a risk-neutral pricing formula under the minimal martingale measure $\widehat \P$ for any $\widehat \P$-integrable European contingent claim. Since the martingale measure is fixed, the risk-neutral price agrees with
the arbitrage free price for
those options which can be replicated by investing on the underlying market. %Denote by $\espp{\cdot}$ the expectation computed under $\widehat \P$ 
Recall that
$X_{t,T}^\tau=X_T^\tau-X_t^\tau$, for each $t \in [0,T)$, refers to the variation of the process $X^\tau$ defined in \eqref{eq:int_info}, over the interval $[t,T]$. Then, 
denote by $\espbar{\cdot\Big{|}\widetilde \F_t}$ the conditional expectation with respect to $\widetilde \F_t$ under the probability measure $\widehat \P$ and so on.
\begin{theorem}\label{th:cont_claim}
Let  $H=\varphi(S_T)$ be the payoff a European-type contingent claim with date of maturity $T$. Then, the risk-neutral price $\Phi_t(H)$ at time $t$ of $H$ is given by
\begin{equation}\label{eq:gen_option}
\Phi_t(H) = \espp{\psi(t,S_t,X_{t,T}^\tau)\Bigg{|} S_t}, \quad t \in [0,T),
\end{equation}
where %the process $\{\psi(t,S_t,X_{t,T}^\tau),\ t \in [0,T)\}$ %
$\psi: [0,T) \times \R_+ \times \R_+ \longrightarrow \R$  is a Borel-measurable function 
such that
\begin{equation}\label{def:psi}
\psi(t,S_t,X_{t,T}^\tau)=B_t\espp{\frac{1}{B_T}G\left(t,S_t,X_{t,T}^\tau,Y_{t,T}\right)\Bigg{|}\F_t^W \vee \F_{T-\tau}^P},
\end{equation}
for a suitable function $G$ depending on the contract such that $G\left(t,S_t,X_{t,T}^\tau,Y_{t,T}\right)$ is $\widehat \P$-integrable.
\end{theorem}
\begin{proof}
For the sake of simplicity suppose that
$\tau<T$ and set $Y_{t,T}:= \int_t^T\sqrt{P_{u-\tau}}\ud \widehat W_u$, for each $t \in [0,T)$. Then, the risk-neutral price $\Phi_t(H)$ at time $t$ of a European-type contingent claim with payoff $H=\varphi(S_T)$ is given by
\begin{align}
\Phi_t(H) & = B_t \espp{\frac{\varphi(S_T)}{B_T}\bigg{|}\widetilde \F_t}\\
& = B_t \espp{\espp{\frac{\varphi\left(S_te^{\int_t^T r(u)\ud u-\frac{\sigma_S^2}{2}X_{t,T}^\tau+ \sigma_S Y_{t,T}}\right)}{B_T}\left\vert\vphantom{\frac{\varphi\left(S_te^{\int_t^T r(u)\ud u-\frac{\sigma_S^2}{2}X_{t,T}^\tau+ \sigma_S Y_{t,T}}\right)}{B_T}}\right. \F_t^W \vee \F_{T-\tau}^P}\left\vert\vphantom{\frac{\varphi\left(S_te^{\int_t^T r(u)\ud u-\frac{\sigma_S^2}{2}X_{t,T}^\tau+ \sigma_S Y_{t,T}}\right)}{B_T}}\right.\widetilde \F_t}, \label{eq:phi_S}
%& = B_t \espp{\espp{\frac{G(S_t,X_{t,T}^\tau, Y_{t,T})}{B_T}\Bigg{|}\F_t^W \vee \F_{T-\tau}^P}\Bigg{|}\tilde \F_t},
\end{align}
where
$\espp{\cdot\Big{|}\widetilde \F_t}$ denotes the conditional expectation with respect to $\widetilde \F_t$ under the minimal martingale measure $\widehat \P$. More generally, \eqref{eq:phi_S} can be written as
\begin{equation}\label{eq:G}
\Phi_t(H) = B_t \espp{\espp{\frac{G(t,S_t,X_{t,T}^\tau, Y_{t,T})}{B_T}\Bigg{|}\F_t^W \vee \F_{T-\tau}^P}\Bigg{|}\widetilde \F_t},
\end{equation}
%Now, observe that the expression in the numerator inside the conditional expectation in \eqref{eq:phi_S} is of the form $G(S_t,X_{t,T}^\tau, Y_{t,T})$, 
for a suitable function $G$ depending on the contract function $\varphi$.
Since the $(\bF,\P)$-Brownian motion $Z$ driving the factor $P$ is not affected by the change of measure from $\P$ to $\widehat \P$ by the definition of minimal martingale measure, %that is, 
we have that $Z$ is also an $(\bF,\widehat \P)$-Brownian motion independent of $\widehat W$, see \eqref{def:hat_Z}.
Hence, we can apply
the same arguments
used in point (ii) of the proof of Theorem \ref{th:sol}, to get that, for each $t \in [0,T)$, the random variable $Y_{t,T}$ 
conditioned on $\F_{T-\tau}^P$ is Normally distributed with mean $0$ and variance $X_{t,T}^\tau$. Then, we can write (in law) that $Y_{t,T}= \sqrt{X_{t,T}^\tau} \epsilon$, where $\epsilon$ is a standard Normal random variable and this allows to find a function $\psi$ such that \eqref{def:psi} holds, which means that 
%define via \eqref{def:psi} the process $\{\psi(t,S_t,X_{t,T}^\tau),\ t \in [0,T)\}$, whose 
%stochastic behavior
%the function $\psi$ in \eqref{def:psi} 
the conditional expectation with respect to $\F_t^W \vee \F_{T-\tau}^P$ in \eqref{eq:G}
only depends on $S_t$ and $X_{t,T}^\tau$, for every $t \in [0,T)$. %his allows to write
Consequently, the risk-neutral price $\Phi_t(H)$ can be written as
\begin{align}
\Phi_t(H) & = \espp{\psi(t,S_t,X_{t,T}^\tau)\bigg{|}\widetilde \F_t}= \espp{\psi(t,S_t,X_{t,T}^\tau)\Bigg{|} S_t},
\label{eq:pricing}
\end{align}
where the last equality holds since $S$ is $\widetilde \bF$-adapted and $X_{t,T}^\tau$ is independent of $\widetilde \F_t$, for each $t \in [0,T)$, 
see e.g. \citet[Lemma A.108]{pascucci2011pde}. More precisely, we have
\begin{equation}
\espp{\psi(t,S_t,X_{t,T}^\tau) \bigg{|}\widetilde \F_t} =\espp{\psi(t,S_t,X_{t,T}^\tau) \bigg{|} S_t}=g(S_t),
\end{equation}
where
\begin{equation}
g(s)=\espp{\psi(t,s,X_{t,T}^\tau) \bigg{|} S_t=s},\quad s \in \R_+.
\end{equation}
%with
%$$
%\psi(t,s,X_{t,T}^\tau)=B_t\espp{\frac{1}{B_T}G\left(t,s,X_{t,T}^\tau,Y_{t,T}\right)\Bigg{|}\F_t^W \vee \F_{T-\tau}^P}.
%$$
\end{proof}

\begin{remark}\label{rem:sigma}
It is worth to remark that $\psi(t,S_t,x)$, with $x \in \R_+$, represents the risk-neutral price at time $t \in [0,T)$ of the contract $H=\varphi(S_T)$ in a Black \& Scholes framework, where the constant volatility parameter $\sigma^{BS}$ is defined by 
$$
\sigma^{BS}:=\sigma_S\sqrt{\frac{x}{T-t}}.
$$
This is proved explicitly in Corollary \ref{th:call2} below for the special case of a \textit{plain vanilla} European Call option.
%Hence, if $f_{X_{t,T}^\tau}:\R_+ \longrightarrow \R$ denotes the density function of $X_{t,T}^\tau$ for each $t \in [0,T)$, provided that it exists, the risk-neutral pricing formula \eqref{eq:gen_option} may be written as:
%\begin{equation}\label{eq:gen_option2}
%\Phi_t(H) =\int_0^{+\infty} \psi(t,S_t,x) f_{X_{t,T}^\tau}(x) \ud x, \quad t \in [0,T).  
%\end{equation}
\end{remark}

\begin{remark}
A pricing formula analogous to \eqref{eq:pricing} is conjectured in \citet{hull1987pricing} for a special example of the model suggested here ($\tau=0$ and $\sigma_S=1$). As already noticed the authors start from the very beginning under a risk neutral framework. Theorem \ref{th:cont_claim} extends their results to the more general case and give a rigorous proof.
\end{remark}

\subsection{A Black \& Scholes-type option pricing formula}

Let us consider a European Call option with strike price $K$ and maturity $T$ and define the function $C^{BS}$ as follows
% In this case, the function $\psi$ in
%the stochastic quantity $\psi(t,S_t,X_{t,T}^\tau)$ appearing in 
%Theorem \ref{th:cont_claim} is simply given by 
\begin{equation} \label{def:pricing_function}
C^{BS}(t,s,x):=s\mathcal N(d_1(t,s,x)) - Ke^{-\int_0^t r(u) \ud u}\mathcal N(d_2(t,s,x)),
\end{equation}
where
\begin{equation} \label{def:d1}
d_1(t,s,x)=\frac{\log\left(\frac{s}{K}\right) + \int_0^t r(u)\ud u + \frac{\sigma_S^2}{2}x}{\sigma_S \sqrt{x}}
\end{equation}
and $d_2(t,s,x)=d_1(t,s,x)-\sigma_S \sqrt{x}$, or more explicitly
\begin{equation} \label{def:d2}
d_2(t,s,x)=\frac{\log\left(\frac{s}{K}\right) + \int_0^t r(u)\ud u - \frac{\sigma_S^2}{2}x}{\sigma_S \sqrt{x}}.
\end{equation}
%\begin{equation} \label{def:d1}
%d_1(t,S_t,X_{t,T}^\tau)=\frac{\log\left(\frac{S_t}{K}\right) + \int_t^T r(u)\ud u + \frac{\sigma_S^2}{2}X_{t,T}^\tau}{\sigma_S \sqrt{X_{t,T}^\tau}}
%\end{equation}
%and $d_2(t,S_t,X_{t,T}^\tau)=d_1(t,S_t,X_{t,T}^\tau)-\sigma_S \sqrt{X_{t,T}^\tau}$, or more explicitly
%\begin{equation} \label{def:d2}
%d_2(t,S_t,X_{t,T}^\tau)=\frac{\log\left(\frac{S_t}{K}\right) + \int_t^T r(u)\ud u - \frac{\sigma_S^2}{2}X_{t,T}^\tau}{\sigma_S \sqrt{X_{t,T}^\tau}}.
%\end{equation}
%Furthermore, 
Here, $\mathcal N$ stands for the standard Gaussian cumulative distribution function
$$
\mathcal N(y)=\frac{1}{\sqrt{2\pi}} \int_{-\infty}^{y} e^{-\frac{z^2}{2}}\ud z, \quad \forall\ y \in \R.
$$
\begin{corollary} \label{th:call2}
The risk-neutral price $C_t$ at time $t$ of a European Call option written on the BitCoin with price $S$ expiring in $T$ and with strike price $K$ is given by the formula
\begin{equation}\label{eq:call}
C_t=\espp{C^{BS}(t,S_t,X_{t,T}^\tau)\bigg{|} S_t}, \quad t \in [0,T), %f_{\widehat X^\tau}(x) \ud x,
\end{equation}
where 
the function $C^{BS}:[0,T) \times \R_+ \times \R_+ \longrightarrow \R$ is given by 
\eqref{def:pricing_function}
and the functions $d_1$, $d_2$ are respectively given by \eqref{def:d1}-\eqref{def:d2}. 
\end{corollary}

\begin{proof}
As in the proof of Theorem \ref{th:cont_claim}, let us assume that %
$\tau<T$. Under the 
minimal martingale measure $\widehat \P$,
the risk-neutral price $C_t$ at time $t \in [0,T)$ of a European Call option written on the BitCoin with price $S$ expiring in $T$ and with strike price $K$, is given by
\begin{align*}
C_t & = B_t \espp{\frac{\max \left( S_T-K,0\right)}{B_T}\bigg{|}\widetilde \F_t}\\
& = B_t\espp{\widetilde S_T\I_{\{S_T>K\}}\Big{|}\widetilde \F_t}-K e^{-\int_t^T r(u) \ud u}\espp{\I_{\{S_T>K\}}\Big{|}\widetilde \F_t}\\
& = B_t J_1 - Ke^{-\int_t^T r(u) \ud u}J_2,
\end{align*}
where we have set $J_1:=\espp{\widetilde S_T\I_{\{S_T>K\}}\Big{|}\widetilde \F_t}$ and $J_2:=\espp{\I_{\{S_T>K\}}\Big{|}\widetilde \F_t}$.
Recall that $Y_{t,T}= \int_t^T\sqrt{P_{u-\tau}}\ud \widehat W_u$, for every $t \in [0,T)$. Then, the term 
$J_2$ can be written as
\begin{align}
J_2 & = \espp{\espp{\I_{\{S_T>K\}}|\F_t^W \vee \F_{T-\tau}^P}\Big{|}\widetilde \F_t}\nonumber \\
& = 
\espp{\widehat \P\left(S_te^{\int_t^T r(u)\ud u-\frac{\sigma_S^2}{2}X_{t,T}^\tau+\sigma_S Y_{t,T}}>K \bigg{|}\F_t^W \vee \F_{T-\tau}^P \right)\bigg{|}\widetilde \F_t}\\
& = 
\espp{\widehat \P\left(-\frac{Y_{t,T}}{
\sqrt{X_{t,T}^\tau}} < \frac{\log\left(\frac{S_t}{K}\right)+ \int_t^T r(u)\ud u - \frac{\sigma_S^2}{2}X_{t,T}^\tau}{
\sigma_S \sqrt{X_{t,T}^\tau}} \bigg{|}\F_t^W \vee \F_{T-\tau}^P\right) \bigg{|}\widetilde \F_t}\\
& = 
\espp{\mathcal N\left( d_2(t,S_t,X_{t,T}^\tau) \right) \bigg{|}\widetilde \F_t}, \end{align}
as for each $t \in [0,T)$,
the random variable $\ds -\frac{Y_{t,T}}{\sqrt{X_{t,T}^\tau}}$ has a standard Gaussian law $\mathcal N(0,1)$ given $\F_t^W \vee \F_{T-\tau}^P$ under the minimal martingale measure $\widehat \P$.
Concerning $J_1$, 
consider the auxiliary probability measure $\bar \P$ on $(\Omega,\F_T)$ 
defined %in \eqref{def:barP}.
as follows:
\begin{equation} \label{def:barP}
\frac{\ud \bar \P}{\ud \widehat \P} := 
e^{-\frac{\sigma_S^2}{2} \int_0^T P_{u-\tau}\ud u + \sigma_S \int_0^T \sqrt{P_{u-\tau}}\ud \widehat W_u}, \quad \widehat \P-\mbox{a.s.}.
\end{equation}
By Girsanov's Theorem, we get that the process $\bar W=\{\bar W_t,\ t \in [0,T]\}$, given by
\begin{equation} \label{def:barW}
\bar W_t := \widehat W_t - \sigma_S\int_0^t\sqrt{P_{u-\tau}}\ud u, \quad t \in [0,T],
\end{equation}
follows a standard $(\bF,\bar \P)$-Brownian motion. In addition, using \eqref{def:tilde_S}, we obtain
\begin{equation} \label{def:tildeST}
\widetilde S_T = \widetilde S_t e^{\sigma_S\int_t^T\sqrt{P_{u-\tau}}\ud \bar W_u + \frac{\sigma_S^2}{2}\int_t^T P_{u-\tau} \ud u},
\end{equation}
for every $t \in [0,T]$. Since $S$ is $\widetilde \bF$-adapted, by \eqref{def:tilde_S} and the Bayes formula on the change of probability measure for conditional expectation, for every $t \in [0,T)$ we get
\begin{align}
J_1 & = \espp{\widetilde S_T \I_{\{S_T>K\}}\Big{|}\widetilde \F_t}\nonumber\\
& = \widetilde S_t \frac{\espp{e^{-\frac{\sigma_S^2}{2} X_{T}^\tau + \sigma_S Y_{0,T}}\I_{\{S_T>K\}}\bigg{|}\widetilde \F_t}}{e^{-\frac{\sigma_S^2}{2} X_{t}^\tau + \sigma_S Y_{0,t}}}\nonumber\\
& = \widetilde S_t \espbar{\I_{\left\{\widetilde S_T > K B_T^{-1}\right\}}\bigg{|}\widetilde \F_t}\nonumber\\
& = \widetilde S_t \espbar{ \espbar{\I_{\left\{
 \sigma_S \bar Y_{t,T} > \log\left(\frac{K}{S_t}\right)- \int_t^T r(u)\ud u - \frac{\sigma_S^2}{2}X_{t,T}^\tau
\right\}}\bigg{|}\F_t^W \vee \F_{T-\tau}^P} \bigg{|}\widetilde \F_t}\nonumber\\
& = \widetilde S_t \espbar{\bar \P \left(
-\frac{\bar Y_{t,T}}{\sqrt{X_{t,T}^\tau}} < \frac{\log\left(\frac{S_t}{K}\right) + \int_t^T r(u)\ud u + \frac{\sigma_S^2}{2}X_{t,T}^\tau}{\sigma_S \sqrt{X_{t,T}^\tau}}
\bigg{|}\F_t^W \vee \F_{T-\tau}^P\right) \bigg{|}\widetilde \F_t}\nonumber\\
& = \widetilde S_t\espbar{\mathcal N\left(d_1(t,S_t,X_{t,T}^\tau) \right) \bigg{|}\widetilde \F_t}, \label{term2}%\\
\end{align}
with 
$$
d_1(t,S_t,X_{t,T}^\tau) = d_2(t,S_t,X_{t,T}^\tau) + \sigma_S\sqrt{X_{t,T}^\tau}.
$$
%where equality \eqref{eq:J1} still follows from \citet[Lemma A.108]{pascucci2011pde}.
In the above computations, analogously to before, we have set $\bar Y_{t,T}:= \int_t^T\sqrt{P_{u-\tau}}\ud \bar W_u$, for each $t \in [0,T)$. Consequently, we have  that $\bar Y_{t,T}$ 
conditional on $\F_{T-\tau}^P$, is a Normally distributed random variable with mean $0$ and variance $X_{t,T}^\tau$, for each $t \in [0,T)$, since $Z$ is not affected by the change of measure from $\widehat \P$ to $\bar \P$.
Indeed,
by the change of numéraire theorem, we have that the probability measure $\bar \P$ turns out to be the minimal martingale measure corresponding to the choice of the BitCoin price process as benchmark.
Further, by applying again
the Bayes formula on the change of probability measure for conditional expectation, we get
\begin{align}
J_1 & = \widetilde S_t\espbar{\mathcal N\left(d_1(t,S_t,X_{t,T}^\tau) \right)\bigg{|}\widetilde \F_t}\nonumber\\
& = \widetilde S_t \frac{\espp{\N\left(d_1(t,S_t,X_{t,T}^\tau) \right)e^{-\frac{\sigma_S^2}{2} \int_0^T P_{u-\tau}\ud u + \sigma_S \int_0^T \sqrt{P_{u-\tau}}\ud \widehat W_u}\bigg{|}\widetilde \F_t}}{e^{-\frac{\sigma_S^2}{2} \int_0^t P_{u-\tau}\ud u + \sigma_S \int_0^t \sqrt{P_{u-\tau}}\ud \widehat W_u}}\nonumber\\
& = \widetilde S_t \espp{\N\left(d_1(t,S_t,X_{t,T}^\tau) \right)e^{-\frac{\sigma_S^2}{2} X_{t,T}^\tau}\espp{e^{\sigma_S Y_{t,T}}
\bigg{|}\F_t^W \vee \F_{T-\tau}^P}\bigg{|}\widetilde \F_t}\nonumber\\
& = \widetilde S_t\espp{\mathcal N\left(d_1(t,S_t,X_{t,T}^\tau) \right)\bigg{|}\widetilde \F_t}, \label{term1}
\end{align}
since the conditional Gaussian distribution of $Y_{t,T}$ gives
$$
\espp{e^{\sigma_S Y_{t,T}}\bigg{|}\F_t^W \vee \F_{T-\tau}^P}=
e^{\frac{\sigma_S^2}{2} X_{t,T}^\tau}.
$$
Finally, gathering the two terms \eqref{term1} and \eqref{term2}, for every $t \in [0,T)$
we obtain
\begin{align}
C_t 
&= S_t\espp{\mathcal N\left(d_1(t,S_t,X_{t,T}^\tau) \right)\bigg{|}\widetilde \F_t}-K e^{-\int_t^T r(u) \ud u}\espp{\mathcal N\left( d_2(t,S_t,X_{t,T}^\tau) \right) \bigg{|}\widetilde \F_t} \nonumber \\
&=\espp{C^{BS}(t,S_t,X_{t,T}^\tau)\bigg{|}\widetilde \F_t}\nonumber \\
&= \espp{C^{BS}(t,S_t,X_{t,T}^\tau)\bigg{|} S_t}, 
\end{align}
where the last equality follows again from \citet[Lemma A.108]{pascucci2011pde}, since
for each $t \in [0,T)$, $X_{t,T}^\tau$ is independent of $\widetilde \F_t$  and $S_t$ is  $\widetilde \F_t$-measurable.

\end{proof}

It is worth noticing that the option pricing formula \eqref{eq:call} only depends on the distribution of $X_{t,T}^\tau$  which is the same both under measure $\widehat \P$ and $\bar \P$. %If $f_{X_{t,T}^\tau}(x)$ denotes the density of this distribution, provided that it exists, 
As observed in Remark \ref{rem:sigma}, formula \eqref{eq:call} evaluated in $S_t$ corresponds to the Black \& Scholes price at time $t \in[0,T)$ of a European Call option written on $S$, with strike price $K$ and maturity $T$, in a market where the volatility parameter is given by $\sigma_S \sqrt{\frac{x}{T-t}}$.
Then, for every $t \in [0,T)$ it  may be written as:
\begin{equation} \label{eq:integr}
C_t=\int_0^{+\infty} C^{BS}(t,S_t,x) f_{X_{t,T}^\tau}(x) \ud x,  %\label{callsimple}
\end{equation}
where %we recall that 
$f_{X_{t,T}^\tau}(x)$ denotes the density function of $X_{t,T}^\tau$, for each $t \in [0,T)$ (if it exists). 
The price at time $t$ for a \textit{plain vanilla} European option may also be written as a Black \& Scholes style price:
\begin{equation} 
C_t=S_t Q_1-K e^{-\int_t^T r(u) \ud s} Q_2, 
\end{equation}
where 
\begin{align} 
Q_1:=\espp{\mathcal N\left(d_1(t,S_t,X_{t,T}^\tau) \right)\bigg{|}S_t}
&= \int_0^{+\infty} \mathcal N\left( d_1(t,S_t,x) \right) f_{X_{t,T}^\tau}(x) \ud x,\\
\end{align}

and \begin{align} 
Q_2:=\espp{\mathcal N\left(d_2(t,S_t,X_{t,T}^\tau) \right)\bigg{|}S_t}
&= \int_0^{+\infty} \mathcal N\left( d_2(t,S_t,x) \right) f_{X_{t,T}^\tau}(x) \ud x.\\
\end{align}
To compute numerically derivative prices by the above formulas, we should compute the  distribution of $X_{t,T}^\tau$, which is not an easy task.

Similar formulas can be computed for other European style derivatives as for binary options which, indeed, are quoted in BitCoin markets. For the case of a Cash or Nothing Call, which is essentially a bet of $A$ on the exercise event, the risk-neutral pricing formula is given by
\begin{align}
C_t^{Bin} & =Ae^{-\int_t^T r(u) \ud s}\espp{\mathcal N\left( d_2(t,S_t,X_{t,T}^\tau) \right)\bigg{|} S_t}\\
&= Ae^{-\int_t^T r(u) \ud s} \int_0^{+\infty} \mathcal N\left( d_2(t,S_t,x) \right) f_{X_{t,T}^\tau}(x) \ud x,\quad t \in [0,T).  \label{binarysimple}
\end{align}
%
%By applying the scaling property of Brownian motion (see e.g. \citet{carr2004bessel}), for every $t \in [0,T]$, we get
%\begin{align*}
%\int_0^t P_u \ud u & =P_0\int_0^t e^{\sigma Z_u+ \left(\mu-0.5 \sigma^2\right) u } \ud u = \frac{4 P_0}{\sigma^2}\int_0^{\sigma^2 t/4 }e^{2 \left(Z_v+ (0.5 \mu-0.25 \sigma^2\right) 4/ \sigma^2 v}\ud v=\frac{4 P_0}{\sigma^2} A_h^{(m)},
%\end{align*}
%where $A_h^{(m)}=\int_0^h e^{2(Z_v+ m v )}\ud v$ is the so-called Yor process with $h=\sigma^2 t/4 $ and $m= 2\mu/\sigma^2-1$.
%The distribution of $X_{0,T}^\tau$ can be thus obtained through the distribution of the process $A_h^{(m)}$, for $h\geq 0$; a rich literature is devoted to this aim, starting from \citet{yor1992some}, e.g. \citet{dufresne2001integral,matsumoto2005exponential,matsumoto2005exponential2}. 
%The main financial application of the above outcomes is to Asian options pricing: see, among others, the seminal paper \citet{geman1993bessel} and more recently \citet{carr2004bessel}.
%A rigorous application of the outcomes of the quoted papers should provide a pricing formula by computing (at least numerically) the integral in \eqref{eq:integr} or in \eqref{binarysimple} %Theorem \ref{th:call2} 
%but this is beyond the scope of this paper. Several approximations have been given in the literature to the distribution of the integral of a Geometric Brownian motion, among others \citet{levy1992pricing} and \citet{MilPosner}.
%

By applying the Levy approximation, see \citet{levy1992pricing}, to $X_0^T$, the Call option pricing formula becomes
\begin{equation}
C_0 = \int_0^{+\infty} C^{BS}(0,S_0,x) \mathcal{LN} pdf_{\alpha(T-\tau),\nu^2(T-\tau)} \left(x\right) \ud x,  \label{callexample}
\end{equation}
which can be computed numerically, once parameters $\alpha(T-\tau),\nu(T-S)$ are obtained. Similarly, Binary Options with terminal value $A$ when in the money, are priced by computing numerically the following formula:
\begin{equation}
C_0^{Bin}=A \int_0^{+\infty} \mathcal N\left( d_2(T-t,S_0,x) \right) \mathcal{LN} pdf_{\alpha(T-\tau),\nu^2(T-\tau)} \left(x\right)\ud x.  
\label{digitalexample}
\end{equation}

\subsection{Numerical exercise}
In this subsection we compute  European plain vanilla and binary option prices assuming that model parameter are known and considering several strike prices and expiration dates. Besides, we also let the initial sentiment and the delay values change in order to understand their contribution to the option price formation. Assume that model parameters are $\mu_P=0.03,\sigma_P=0.35, \sigma_S=0.04$, the riskless interest rate is $r=0.01$ and that the BitCoin price at time $t=0$ is $S_0=450$ (this is the price by October 2016).

In Table \ref{tab:prices},  Call option prices are reported for $T=3$ months, $\tau$= 5 days. Rows correspond to different values of $P_0$ while columns to different values for the strike price. As expected, Call option prices are increasing with respect to initial sentiment for the BitCoin and decreasing with respect to strike prices.
\begin{table}[htp]
\caption{Call option prices against different strikes $K$ and for different values of $P_0$: $S_0=450,r=0.01,\mu_P=0.03,\sigma_P=0.35, \sigma_S=0.04$, $T=3$ months, $\tau$ = $1$ week (5 days).}
 \label{tab:prices}
 \centering
\begin{tabular}{||c|c|c|c|c|c||}
\hline 
K & 400 & 425 & 450 & 475 & 500 \\ 
\hline 
$P_0=10$ & 51.24 & 28.35 & 11.46 & 3.09 & 0.54 \\ 
\hline 
$P_0=100$ & 64.12 & 48.05 & 34.94 & 24.69 & 16.97 \\ 
\hline 
$P_0=1000$ & 128.68 & 117.75 & 107.77 & 98.66 & 90.35 \\ 
\hline
\end{tabular}
\end{table}

In Table \ref{tab:prices2}, Call option prices are summed up, for $P_0=100$, by letting the expiration date $T$ and the delay $\tau$ vary. Again as expected, for Plain Vanilla Calls the price increases with time to maturity.
Increasing the delay reduces option prices; of course the spread is inversely related to the time to maturity of the option. 

\begin{table} [htb]
 \caption{Call option prices against different Strikes $K$ and for different values of $T$ and $tau$: $S_0=450,r=0.01,\mu_P=0.03,\sigma_P=0.35, \sigma_S=0.04$ and $P_0=100$.}
\label{tab:prices2}
\centering
\begin{tabular}{||c|c|c|c|c|c||}

\hline 
K & 400 & 425 & 450 & 475 & 500 \\ 
\hline
$T$=1 month, $\tau$=1 week & 52.85 & 33.09 & 18.27 & 8.81 & 3.71 \\ 
\hline 
$T$=1 month, $\tau$=2 weeks & 51.58 & 30.62 & 15.18 & 6.13 & 2.00 \\ 
\hline 
$T$=3 months, $\tau$=1 week  & 64.12 & 48.05 & 34.94 & 24.69 & 16.97 \\
\hline
$T$=3 months, $\tau$=2 weeks & 62.95 & 46.65 & 33.42 & 23.18 & 15.60 \\ 
\hline
 \end{tabular}
 \end{table}

In Tables \ref{tab:prices3} and \ref{tab:prices4}, analogous results are reported for Binary Options with outcome $A=100$; Table \ref{tab:prices3}  sums up Binary Cash-or-Nothing prices for $S_0=450$, $r=0.01$, $\mu_P=0.03$, $\sigma_P=0.35$, $\sigma_S=0.04$, $T=3$ months, $\tau= 1$ week (5 working days) against several strikes (in colums). Rows correspond to different values  $P_0$ for the initial sentiment on BitCoins. As expected, prices are decreasing with respect to strike prices. Here, \textit{in the money} (ITM) options values are decreasing with respect to $P_0$ while \textit{out of the money} (OTM) ones are increasing. The difference in ITM and OTM prices is large for low values of $P_0$, while it is very small for a high level of the initial sentiment factor in BitCoins. This may be justified by the fact that, when the sentiment factor in the BitCoin is strong, all bets are worth, even the OTM ones, since the underlying value is expected to blow up.
Binary Call prices decrease with respect to time to maturity for ITM options and increase for OTM options which become more likely to be exercised. The influence of the delay value is tiny, as for vanilla options, being larger for short time to maturities.

\begin{table}[htb]
\caption{Digital Cash or Nothing prices against different Strikes $K$ and for different values of $P_0$ on BitCoins. Market parameters are $S_0=450$, $r=0.01$, $\mu_P=0.03$, $\sigma_P=0.35$, $\sigma_S=0.04$, $T=3$ months, $\tau= 5$ days. The prize of the option is set to $A=100$.}
\label{tab:prices3}
\centering
\begin{tabular}{||c|c|c|c|c|c||}
\hline 
K & 400 & 425 & 450 & 475 & 500 \\ 
\hline 
$P_0=10$ & 97.17 & 82.77 & 50.31 & 18.87 & 4.24 \\ 
\hline 
$P_0=100$ & 70.07 & 58.38 &46.58 & 35.66 & 26.27 \\ 
\hline 
$P_0=1000$ & 45.70& 41.77 & 38.14 & 34.79 & 31.72 \\ 
\hline
\end{tabular}
 \end{table}

\begin{table} [htb]
 \caption{Digital Cash or Nothing prices against different Strikes $K$ and for different values of $T$ and $\tau$. Market parameters are $S_0=450$, $r=0.01$, $\mu_P=0.03$, $\sigma_P=0.35$, $\sigma_S=0.04$ and $P_0=100$. The prize of the option is set to $A=100$.}
\label{tab:prices4}
\centering 
\begin{tabular}{||c|c|c|c|c|c||}
\hline 
K & 400 & 425 & 450 & 475 & 500 \\ 
\hline 
$T$=1 month, $\tau$=1 week & 86.93 & 69.97 & 48.27 & 28.11 & 13.83 \\ 
\hline 
$T$=1 month, $\tau$=2 weeks & 91.50 & 74.23 & 48.69 & 24.84 & 9.80 \\ 
\hline
$T$=3 months, $\tau$=1 week & 70.07 & 58.38 & 46.58 & 35.66 & 26.27 \\ 
\hline
$T$=3 months, $\tau$=2 weeks & 71.21 & 59.10 & 46.77 & 35.36 & 25.62 \\ 
\hline 

\end{tabular}
\end{table}

\section{Model fitting on real data} \label{sec:fit}

This section is devoted to the estimation of the model in (\ref{eq:S})-(\ref{eq:BSdyn}) on real data; the overall procedure is aimed at describing the dynamics of BitCoin price changes over time. In order to fit the model we need data for both the BitCoin price and the sentiment indicator. Several proxies have been suggested for the latter; traditional indicators of market sentiment on a stock asset such as the number and volume of transactions \citep{MainDrivers} as well as sentiment indicators such as the number of Google searches or Wikipedia requests in the period under investigation \citep{bukovina2016sentiment}. Internet-based proxies are particularly interesting for the BitCoin price formation, being BitCoin itself an internet-based asset.
First we investigate which of the suggested proxies is consistent with the dynamics in \eqref{eq:BSdyn} and then we fit the full model to BitCoin prices. Daily data for BitCoin prices, volume and number of transactions are obtained through the website \emph{http://blockchain.info} which provides a mean price among main exchanges trading on BitCoin and the total exchanged volume.  Weekly data for the number of Google searches are downloaded from Google-Trends website (daily data are not available). Daily data for Wikipedia requests are obtained through the website \emph{http://tools.wmflabs.org/pageviews}.

\subsection{Proxies for Sentiment}\label{sec:confind} 
The univariate process  $P_t$ is a Geometric Brownian motion; the corresponding discrete process of logarithmic returns is given by $X:={X_i,i\in \bN}$, where $X_i=\log\left(\frac{P_i}{P_{i-1}}\right)$,  and it is well-known that these are independent an identically distributed with mean $\left(\mu_P-\frac{\sigma_P^2}{2}\right) \delta$ and variance $\sigma_P^2 \delta$.
Hence, in order to choose a suitable proxy, based on discrete observations, for the process $P$, we simply perform a stationary test and a normality test on the corresponding realizations $\lbrace x_i\rbrace_i$ of $X$, using respectively the augmented Dickey–Fuller test \citep{Tsay:timeseries}, and the one-sample Kolmogorov-Smirnov test \citep{Massey:kstest}. 
We consider the number and the volume of transactions as examples of traditional indicators and the number of Google searches and Wikipedia requests as examples of sentiment indicators. The first three series were investigated from 01/01/2012 to 31/03/2017 while Wikipedia requests are considered from 01/07/2015 to 31/03/2017. 

The tests are performed on the whole time series and on the sub-samples from 01/01/2015 to 31/03/2017.

In Table \ref{tab:stattest} and in Table \ref{tab:logtest} we report the outcomes of the tests.

\begin{table}[htbp]
\caption{Augmented Dickey-Fuller test for $X$ with $\alpha=0.05$}	
\centering
\captionsetup{justification=centering}
\label{tab:stattest}
\begin{tabular}{lrrrr}
	\toprule
	Time series & \multicolumn{4}{c}{p-value\footnotemark[1]} \\ 
	 & Num. trans. & Vol. trans. & Google searches & Wiki requests\\ 
	  \midrule
All series & 1.0000e-03*** & 1.0000e-03*** & 1.0000e-03*** & \multicolumn{1}{c}{-}\\
Sub-sample & 1.0000e-03*** & 1.0000e-03*** & 1.0000e-03*** & 1.0000e-03***\footnotemark[2]\\
	\bottomrule
	\end{tabular} 
\end{table}

\begin{table}[htbp]
\caption{Kolmogorov-Smirnov test for $X$ with $\alpha=0.05$}	
\centering
\captionsetup{justification=centering}
\label{tab:logtest}
\begin{tabular}{lrrrr}
	\toprule
	Time series & \multicolumn{4}{c}{p-value\footnotemark[1]} \\ 
	 & Num. trans. & Vol. trans. & Google searches & Wiki requests\\ 
	  \midrule
All series & 5.1011e-05**** & 4.3576e-07**** & 2.9246e-07**** & \multicolumn{1}{c}{-}\\
Sub-sample & 0.0035**   & 0.2152  & 0.1012  & 1.3575e-11****\footnotemark[2]\\
	\bottomrule
	\end{tabular} 
\end{table}

\footnotetext[1]{
\begin{tabular}{llllllllllll}
	*    & $P \leq 0.05$ & ; & ** & $P \leq 0.01$ & ; &	*** & ; & $P \leq 0.001$ & ; &	**** & $P \leq 0.0001$ \\		
	\end{tabular} 
}

\footnotetext[2]{Data available only from 01/07/2015.}

Non-stationarity is rejected for all proxies; besides, lognormality is not rejected, for the sub-sample from 01/01/2015 to 31/03/2017, both for the $Volume$ and the  number of $Google\; searches$.
In Figure \ref{fig:test_norm} we plot the histograms of the latter two proxies; it is evident that the log-normality fit for the $Volume$ time series is better.  

\begin{figure}[htbp]
\centering
\includegraphics[width=.45\textwidth]{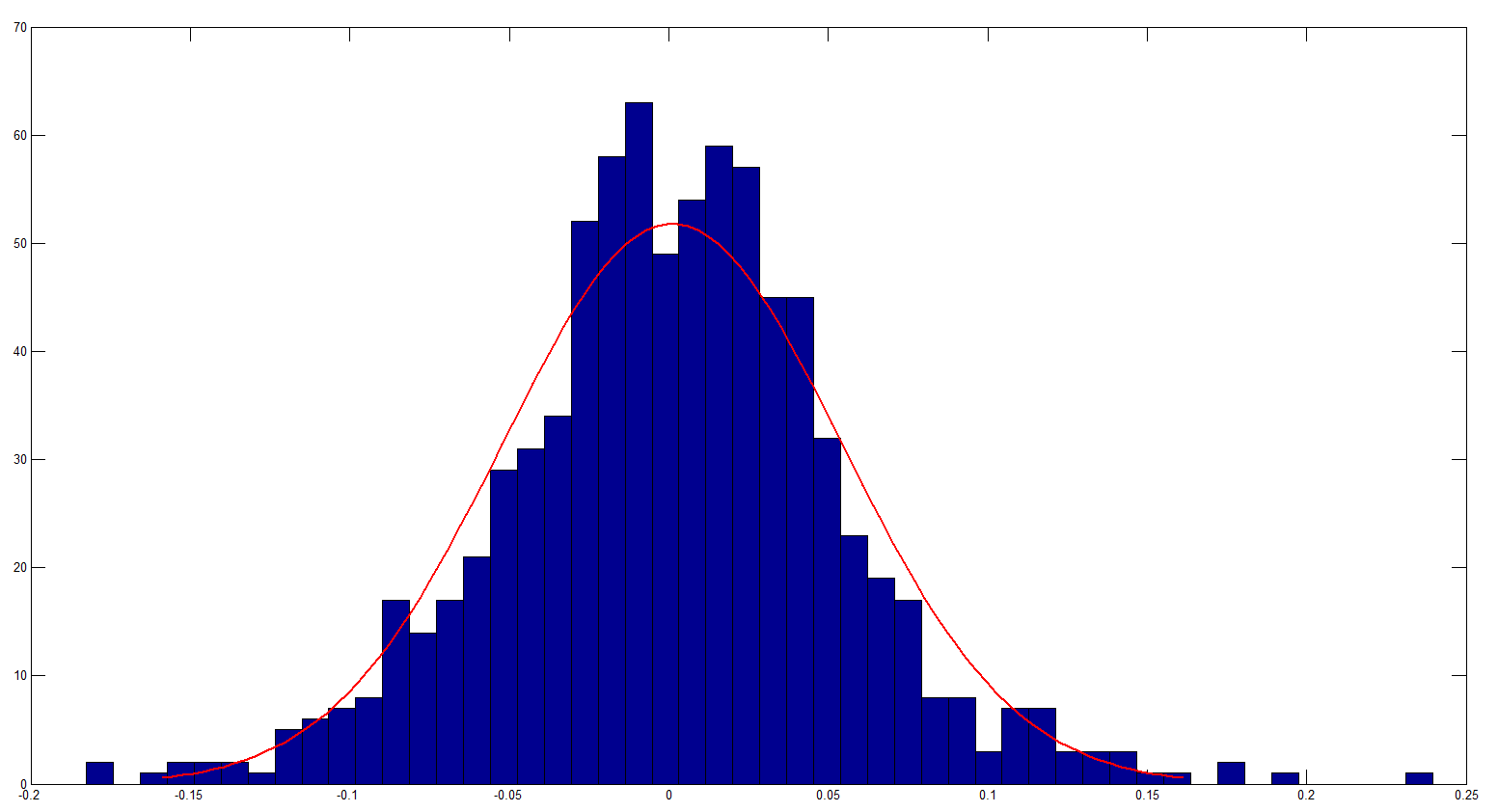}\hfil
\includegraphics[width=.45\textwidth]{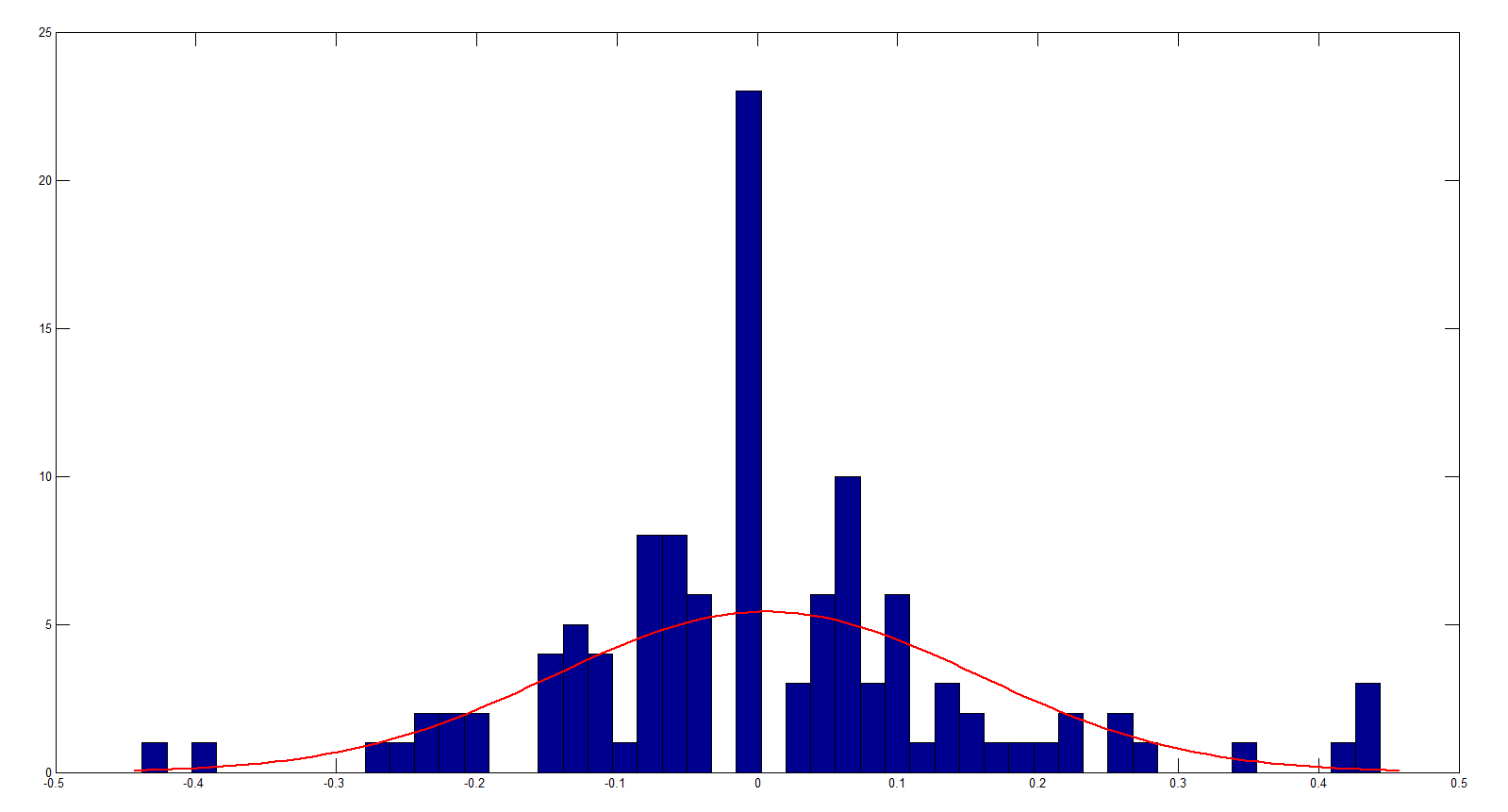}
\caption{Histogram with normal distribution fit for $P=volume\; of\; transactions$ (on the left) and $P=Google\; searches$ (on the right).}\label{fig:test_norm}
\end{figure}

\subsection{Estimation Results with QML method}\label{sec:est_res}
According to the outcomes in Table \ref{tab:logtest} we consider the daily time series of the volume of transactions and the weekly time series of the Google searches from 01/01/2015 to 31/03/2017 as suitable proxies for process $P$. We fit the model in \eqref{eq:S} and \eqref{eq:BSdyn} to the BitCoin price and the volume of transaction as well as to the BitCoin prices and number of Google searches by applying the Profile-Quasi-Likelihood (PML) method as describes in Section \ref{sec:fit}.  Google-trends provides a scaled time series for the number of searches so the the maximum value is $100$; in order to compare outcomes we do the same for the $Volume$ time series. 
%This scaling avoid that too large values of $P$ generate not plausible values for $S$ unless $\mu_s,\sigma_S$ were negligible.
In Figures \ref{fig:path_goog} and \ref{fig:path_vol} we plot the time series of BitCoin price with Google searches (weekly) and the time series of BitCoin price with volume of transactions (daily), respectively.

\begin{figure}[htbp]
\includegraphics[scale=0.38]{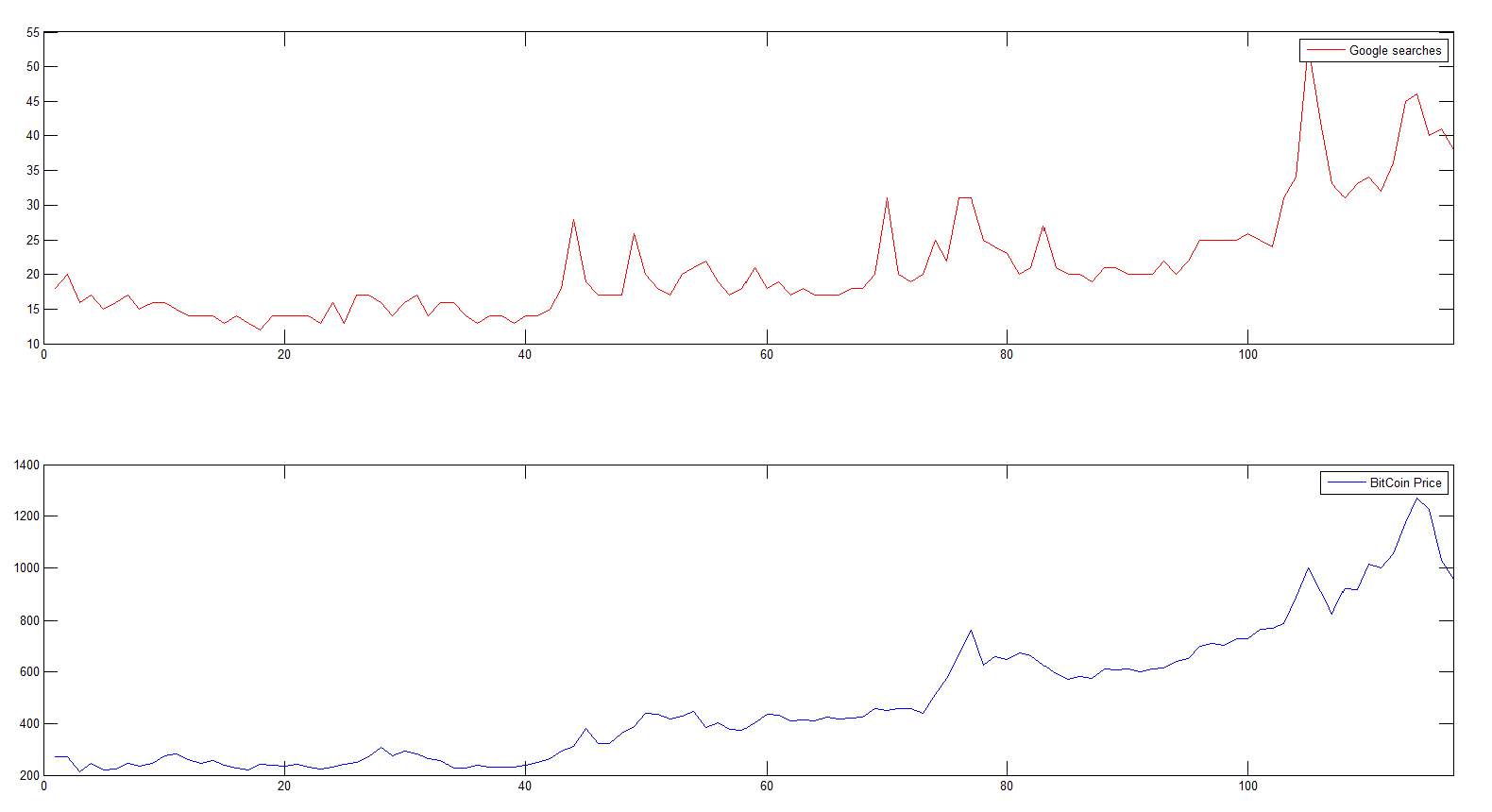} 
\caption{Weekly time series of the Google searches (above) and BitCoin price (bottom) from 01/01/2015 to 31/03/2017.}\label{fig:path_goog}
\end{figure}

\begin{figure}[htbp]
\includegraphics[scale=0.38]{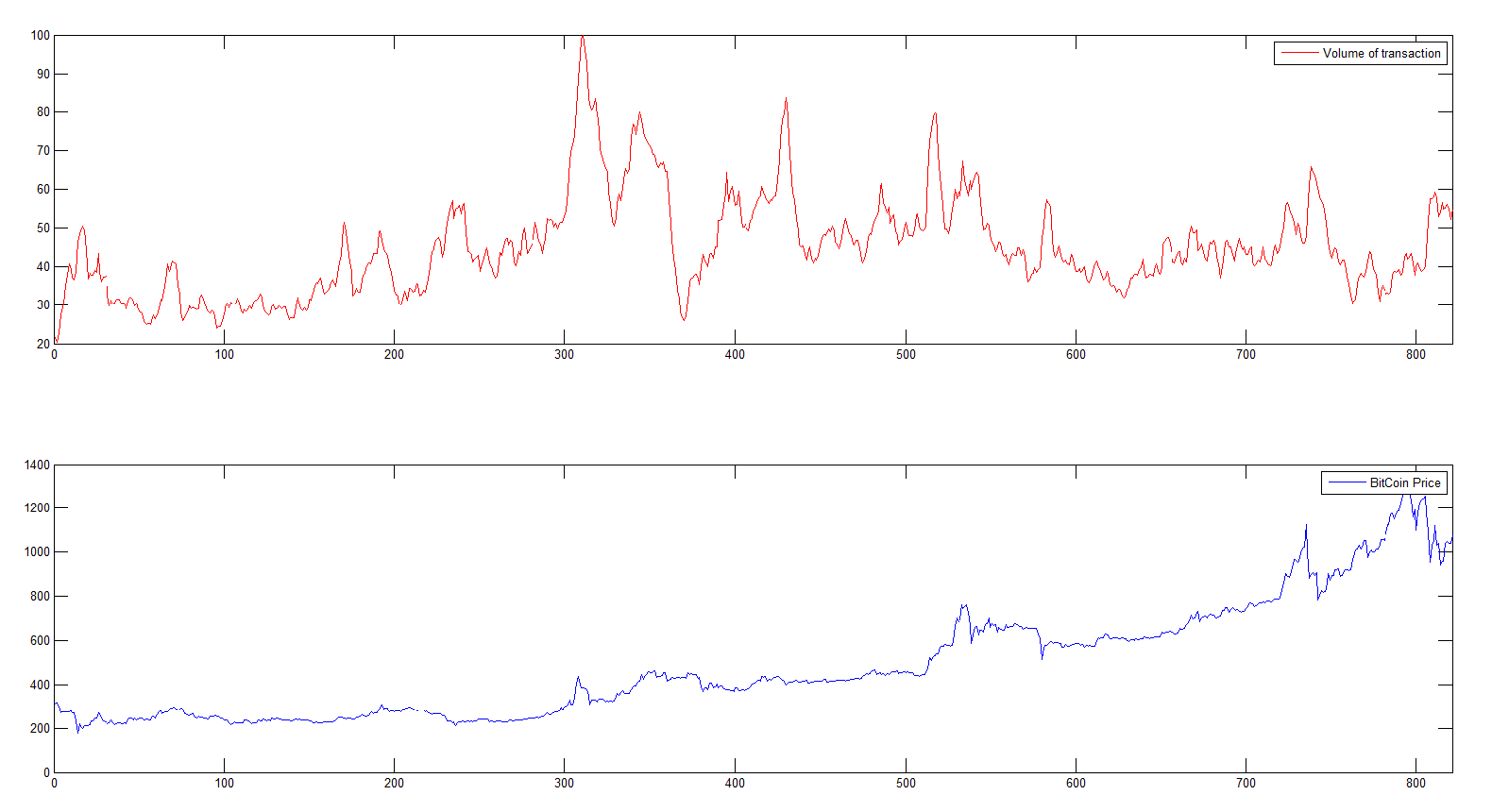} 
\caption{Daily time series of the volume of transactions (above) and BitCoin price (bottom) from 01/01/2015 to 31/03/2017.}\label{fig:path_vol}
\end{figure}

In what follows we assume that $\Delta$=1 week is the observation step for BitCoin log-returns.

\subsubsection{Sentiment maesured by Volume}\label{sec:voltrans}
Given daily observations  $ \lbrace P_i \rbrace _i$ of the volume of transactions we are able to compute the cumulative weekly sentiment $\lbrace A_i \rbrace_i$; for $\tau=0$ $A_i$ is simply the mean volume during the preceding week i.e.
\begin{equation}
\label{eq:apprAi}
A_i=\int_{\left(i-1\right)\Delta}^{i\Delta}P_{u} \ud u =\sum_{j=1}^{7}\int_{\left(i-1\right)\Delta+\left(j-1\right)\delta}^{\left(i-1\right)\Delta+j\delta} P_{u} \ud u =\sum_{j=1}^{7} P_{\left(i-1\right)\Delta+j\delta} \delta=\frac{\sum_{j=1}^{7} P_{\left(i-1\right)\Delta+j\delta}}{7} \Delta,
\end{equation}
The generalization to $\tau>0$ is straightforward as soon as we assume $\tau=r\delta$ for some positive integer $r$; in which case $A_i$ would be the mean volume of the 7 days preceding time $i\Delta-r\delta$.

By applying the profile quasi maximum likelihood we obtain $\tau=5$ days and the following estimates of other the parameters:

\begin{table}[htbp]
\caption{Parameter fit with $\tau=5$}	
\centering
\captionsetup{justification=centering}
\label{tab:fitparvol2}
\begin{tabular}{rrrrrr}
	\toprule
Variable    &  Fitted value  &  Std. error &  t-value & P($>\lvert t \rvert$)  \\ 
	  \midrule
$\mu_P$		&	1.0404	&	0.7373	&	1.4110	&	0.1610	\\
$\sigma_P$	&	1.1092	&	0.0725	&	15.2924	&	0.0000	\\
$\mu_S$		&	0.0153	&	0.0083	&	1.8434	&	0.0679	\\
$\sigma_S$	&	0.0830	&	0.0054	&	15.2403	&	0.0000	\\
	\bottomrule
	\end{tabular} 
\end{table} 
In order to asses parameters significance the $t$-stat is computed, for each parameter, under the null hypothesis that its value is zero;  Table \ref{tab:fitparvol2} shows that $\mu_P$ is not statistically significant and that $\mu_S$ is weakly significant.
Finally we evaluate the confidence region for $\tau$ using the \eqref{eq:confreg} and $\tau\in \lbrace 0,1,2,\ldots,10 \rbrace$ days. We find that  $\tau=2,3,4$ days belong to the confidence region. Estimates of other parameters are indeed very similar in any of such cases and analogous comments apply.
%
%\begin{table}[htbp]
%\caption{Parameter fit with $\tau=3$ and $\tau=4$}	
%\centering
%\captionsetup{justification=centering}
%\label{tab:fitparvol34}
%\begin{tabular}{rrrrrrr}
%	\toprule
%Variable & \multicolumn{3}{c|}{$\tau=3$} & \multicolumn{3}{|c}{$\tau=4$}\\
%    &  Fitted value & Std. error & P($>\lvert t \rvert$) & \multicolumn{1}{|r}{Fitted value} & Std. error & P($>\lvert t \rvert$)\\ 
%	  \midrule
%$\mu_P$	   & 1.1168 & 0.7520 & 0.1403	& 1.0619 & 0.7538 & 0.1617	\\
%$\sigma_P$ &	1.1250 & 0.0735	& 0.0000	& 1.1154 & 0.0729 & 0.0000	\\
%$\mu_S$	   &	0.0151 & 0.0081	& 0.0653	& 0.0152 & 0.0082 & 0.0666	\\
%$\sigma_S$ &	0.0815 & 0.0053	& 0.0000	& 0.0822 & 0.0054 & 0.0000	\\
%	\bottomrule
%	\end{tabular} 
%\end{table} 

%which are very close to the previous ones so analogous comments apply. In the Appendix \ref{sec:estres_vol} we report the outcomes for all other values of $\tau$.

%Further research will be devoted for high frequency data. 

%\begin{figure}[htbp]
%\includegraphics[scale=0.38]{sim_vol.png} 
%\caption{Real Bitcoin price (solid line) and simulated Bitcoin price (dotted data)}\label{fig:simvol}
%\end{figure}
%
%Finally in figure \ref{fig:simvol} we plot the real Bitcoin price and a selection of possible simulated path by using the parameter found for $\tau=0$ within one year horizon.

\subsubsection{Sentiment measured by Google Searches}
Assume now that Google searches are representative of sentiment about BitCoin. Since we have weekly data we should aggregate both Google searches and BitCoin returns to a coarser observation step; however this would reduce the time series length dramatically and corresponding estimates might be unreliable. Hence we assume that the available observations correspond to the cumulative sentiment time series $A_i$ and we assume it exist a non-negative integer $c$ such that $\tau=c\Delta$ i.e. in this case $\tau$ is on a weekly scale and not on a daily scale like in the previous case. 

By applying the profile quasi maximum likelihood we obtain $\tau=1$ week and the following estimates of other parameters:

\begin{table}[htbp]
\caption{Parameter fit with $\tau=1$}	
\centering
\captionsetup{justification=centering}
\label{tab:fitpargoog1}
\begin{tabular}{rrrrrr}
	\toprule
Variable    &  Fitted value  &  Std. error &  t-value & P($>\lvert t \rvert$)  \\ 
	  \midrule
$\mu_P$		&	0.9573	&	0.7315	&	1.3087	&	0.1935	\\
$\sigma_P$	&	1.0818	&	0.0714	&	15.1611	&	0.0000	\\
$\mu_S$		&	0.0181	&	0.0092	&	1.9534	&	0.0535	\\
$\sigma_S$	&	0.0867	&	0.0057	&	15.1774	&	0.0000	\\
	\bottomrule
	\end{tabular} 
\end{table} 

The $t$-value is reported for all parameters as in Section \ref{sec:voltrans}.  Again,  $\mu_P$ is not statistically significant and that $\mu_S$ is weakly significant.
Finally we evaluate the confidence regions for $\tau$ using the \eqref{eq:confreg} and $\tau\in \lbrace 0,1,2,\ldots,10 \rbrace$ weeks. We find that there are not other values of $\tau$ that belong to the confidence region.

\section{Model performance on market option prices} \label{sec:assessing}
Recently some online platforms have appeared where it is possible to trade on plain vanilla options on the BitCoin. A relevant platform where bid-ask quotes are publicly available is www.deribit.com; we will consider option prices on this website as "market prices".
\begin{figure}[htbp]
\centering
\includegraphics[width=.70\textwidth]{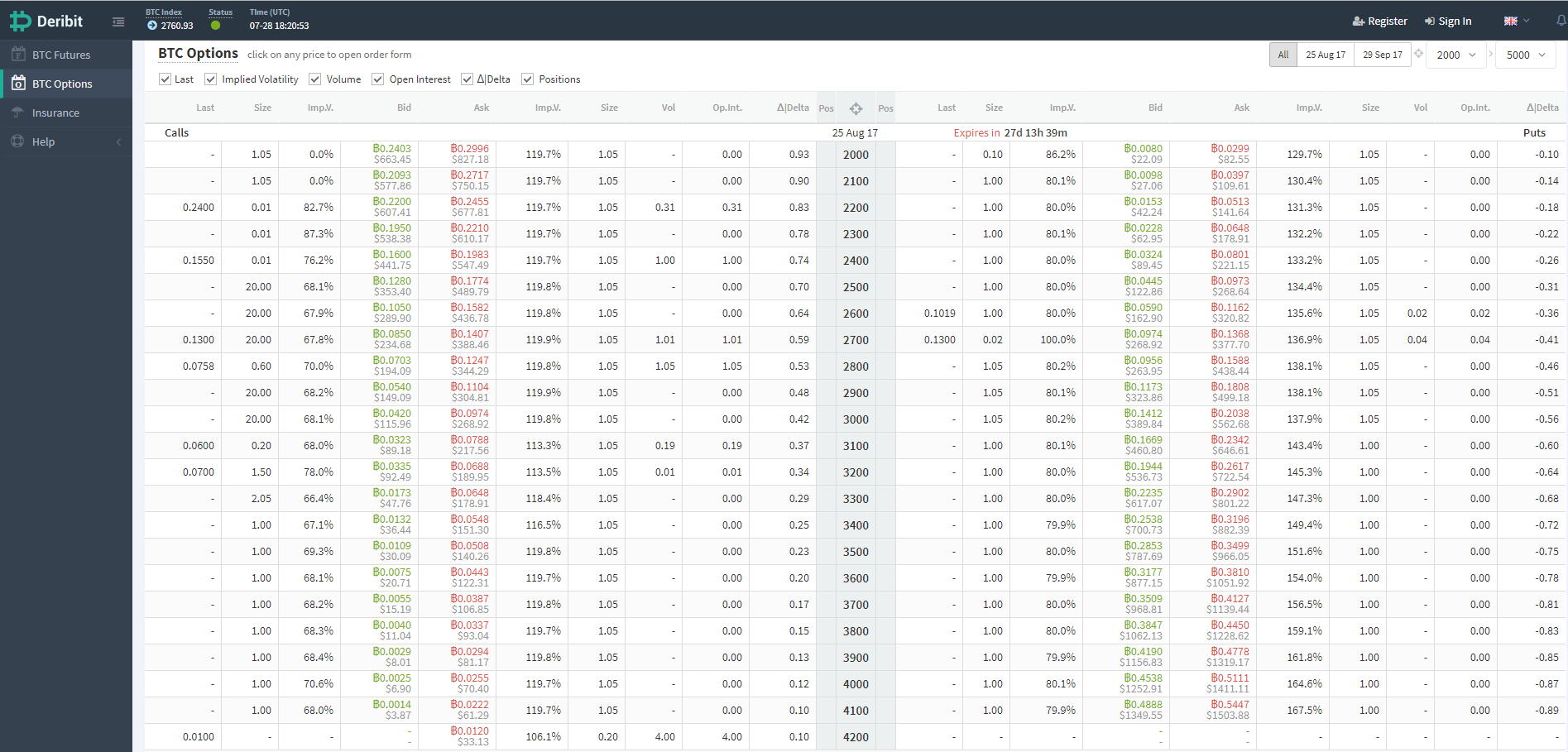}
\caption{Screenshot of the website \url{www.deribit.com} on July 28, 2017}\label{fig:screen}
\end{figure}

In Figure \ref{fig:screen} the screenshot of the website on July 28, 2017; we will consider the mid-value of the Bid-Ask range as a benchmark for assessing model performance discarding options for which there was no transactions. Concerning the price of the underlying BitCoin, it was set as the price of the BitCoin index available from  blockchain.info at the same time as the option data download, which also appears in the north-west corner of the screenshot. Every day two different expiration dates are available corresponding to a one month and two months maturity at issue. 
We are aware of possible synchronicity problems but as a first evaluation of the suggested pricing formula we intentionally neglect this friction.
Model prices are then compared with corresponding market prices by computing the Root Mean Squared Error of the model across all the considered sample of options and of suitably chosen sub-samples. The same is done when prices are computed with the benchmark no-sentiment model chosen as Black and Scholes model. Of course the latter is estimated on the same time-series as those considered in Subsection \ref{sec:est_res}. Only the volatility estimation matters since we set $r=0$ for both models. 

In Tables \ref{tab:optprice28}-\ref{tab:optprice28bis} we report market data for July 28, 2017 as well as model option prices in the suggested model and for the Black and Scholes benchmark. On the chosen day available maturities were complete i.e. August 25 and September 29.

\begin{table}[htbp]
\caption{Option Price with t=28 July and T=25 August}	
\centering
\captionsetup{justification=centering}
\label{tab:optprice28}
\begin{tabular}{lrrrrr}
	\toprule
K    &   Bid  &  Ask   & Model (Vol.) & Model (Google)  &  BS   \\ 
	  \midrule
2200 & 0.2200 &	0.2455 & 0.2125	 	  & 0.2196			& 0.2110\\
2300 & 0.1956 &	0.2210 & 0.1820	      & 0.1909			& 0.1799\\
2400 & 0.1600 &	0.1983 & 0.1538	      & 0.1645			& 0.1510\\
2500 & 0.1280 &	0.1774 & 0.1281	      & 0.1404			& 0.1248\\
2600 & 0.1050 &	0.1582 & 0.1052	      & 0.1187			& 0.1015\\
2700 & 0.0850 &	0.1407 & 0.0867	      & 0.0995			& 0.0812\\
2800 & 0.0703 &	0.1247 & 0.0696	      & 0.0827			& 0.0640\\
2900 & 0.0540 &	0.1104 & 0.0482	      & 0.0698			& 0.0497\\
	\bottomrule
\multicolumn{3}{l}{RMSE mean bid/ask} & 0.0753 & 0.0395 & 0.0833 \\  
	\end{tabular} 
\end{table}

\begin{table}[htbp]
\caption{Option Price with with t=28 July and T=29 September }	
\centering
\captionsetup{justification=centering}
\label{tab:optprice28bis}
\begin{tabular}{lrrrrr}
	\toprule
K    &   Bid   	&  Ask   	& Model (Vol.) & Model (Google)  &  BS   \\ 
	  \midrule
2200 &	0.2108  &	0.2764	&	0.2350		&	0.2981		&	0.2297	\\
2300 &	0.1881	&	0.2547	&	0.2091		&	0.2779		&	0.2029	\\
2400 &	0.2004	&	0.2343	&	0.1851		&	0.2589		&	0.1781	\\
2500 &	0.1700	&	0.2154	&	0.1631		&	0.2411		&	0.1555	\\
2600 &	0.1650	&	0.1950	&	0.1431		&	0.2244		&	0.1349	\\
2700 &	0.1168	&	0.1710	&	0.1249		&	0.2087		&	0.1164	\\
2800 &	0.1044	&	0.1630	&	0.1086		&	0.1941		&	0.0999	\\
2900 &	0.0934	&	0.1490	&	0.0940		&	0.1804		&	0.0853	\\
	\bottomrule
	\multicolumn{3}{l}{RMSE mean bid/ask} & 0.0723 & 0.1535 &   0.0932  \\  
	\end{tabular} 
\end{table} 

In Tables \ref{tab:RMSE_vol2}-\ref{tab:RMSE_goog2} we sum up the outcomes for the RMSE for All options and for subsamples obtained by considering the same expiration date respectively when sentiment is conveyed by the Volume and by Google searches.
The same in Tables \ref{tab:RMSE_vol1}-\ref{tab:RMSE_goog1} where the subsample are obtained according to moneyness. 
Highlighted in the table (in bold) are cases where the plain vanilla Black and Scholes model does better than the model suggested here. Overall our pricing formula does much better than the benchmark in all cases. It is worth noticing that Google Searches tend to overprice long-term options while it is the very best for shorter term options as if this sentiment indicator is driven by enthusiasm giving a sudden impulse to options. 

\begin{table}[htbp]
\caption{RMSE Option Price with Volume of Transactions}	
\centering
\captionsetup{justification=centering}
\label{tab:RMSE_vol1}
\begin{tabular}{lrrr}
	\toprule
	Options 		& Num. &	RMSE Model	&	RMSE BS	\\
	  \midrule
	All 			& 144  &	0.3089		&	0.4879	\\
	Very Shorts 	& 16   &	0.0898		&	\textbf{0.0817}	\\
	1 Months 		& 32   &	0.1132		&	0.1448	\\
	2 Months 		& 32   &	0.1306		&	0.1751	\\
	ITM 			& 54   &	0.1536		&	0.2617	\\
	ATM 			& 36   &	0.1687		&	0.2625	\\
	OTM 			& 54   &	0.2082		&	0.3173	\\
	\bottomrule
	\end{tabular} 
\end{table} 

\begin{table}[htbp]
\caption{RMSE Option Price by Date with Volume of Transactions}	
\centering
\captionsetup{justification=centering}
\label{tab:RMSE_vol2}
\begin{tabular}{lrrrrrr}
	\toprule
 & \multicolumn{2}{c}{Near} & \multicolumn{2}{c}{Next} & \multicolumn{2}{c}{All Day}		\\	
Date	&	RMSE Model	&	RMSE BS	&	RMSE Model	&	RMSE BS	&	RMSE Model	&	RMSE BS	\\
	  \midrule
20/07/2017 & 0.0712	&\textbf{0.0648}&	0.0627	&	0.1768	&	0.0949	&	0.1883	\\
21/07/2017 & 0.0547	&\textbf{0.0499}&	0.0172	&	0.1428	&	0.0574	&	0.1512	\\
22/07/2017 & -		&	-			&	0.0180	&	0.1394	&	0.0180	&	0.1394	\\
23/07/2017 & -		&	-			&	0.0556	&	0.1475	&	0.0556	&	0.1475	\\
24/07/2017 & -		&	-			&	0.0773	&	0.1420	&	0.0773	&	0.1420	\\
25/07/2017 & -		&	-			&	0.1056	&	0.1445	&	0.1056	&	0.1445	\\
26/07/2017 & -		&	-			&	0.1252	&	0.1519	&	0.1252	&	0.1519	\\
27/07/2017 & -		&	-			&	0.1306	&	0.1508	&	0.1306	&	0.1508	\\
28/07/2017 & 0.0753	&	0.0833		&	0.0723	&	0.0932	&	0.1044	&	0.1250	\\
29/07/2017 & 0.0562	&	0.0717		&	0.0493	&	0.0919	&	0.0748	&	0.1165	\\
30/07/2017 & 0.0531	&	0.0710		&	0.0705	&	0.0850	&	0.0883	&	0.1107	\\
31/07/2017 & 0.0343	&	0.0621		&	0.0664	&	0.0795	&	0.0748	&	0.1009	\\
	\bottomrule
	\end{tabular} 
\end{table}

\begin{table}[htbp]
\caption{RMSE Option Price with Google Searches}	
\centering
\captionsetup{justification=centering}
\label{tab:RMSE_goog1}
\begin{tabular}{lrrr}
	\toprule
	Options 		& Num. &	RMSE Model	&	RMSE BS	\\
	  \midrule
	All 			& 136	&	0.3380		&	0.4854	\\
	Very Shorts 	& 8		&	0.0367		&	0.0648	\\
	1 Months 		& 32	&	0.0621		&	0.1448	\\
	2 Months  		& 32	&	0.2408		&	\textbf{0.1751}	\\
	ITM 			& 51	&	0.1912		&	0.2617	\\
	ATM 			& 34	&	0.1613		&	0.2612	\\
	OTM 			& 51	&	0.2273		&	0.3149	\\
	\bottomrule
	\end{tabular} 
\end{table} 

\begin{table}[htbp]
\caption{RMSE Option Price by Date with Google Searches}	
\centering
\captionsetup{justification=centering}
\label{tab:RMSE_goog2}
\begin{tabular}{lrrrrrr}
	\toprule
 & \multicolumn{2}{c}{Near} & \multicolumn{2}{c}{Next} & \multicolumn{2}{c}{All Day}		\\	
Date	&	RMSE Model	&	RMSE BS	&	RMSE Model	&	RMSE BS	&	RMSE Model	&	RMSE BS	\\
	  \midrule
20/07/2017 & 0.0367	&	0.0648	&	0.0578	&	0.1768	&	0.0684	&	0.1883	\\
21/07/2017 & -		&	-		&	0.0158	&	0.1428	&	0.0158	&	0.1428	\\
22/07/2017 & -		&	-		&	0.0253	&	0.1394	&	0.0253	&	0.1394	\\
23/07/2017 & -		&	-		&	0.0891	&	0.1475	&	0.0891	&	0.1475	\\
24/07/2017 & -		&	-		&	0.0935	&	0.1420	&	0.0935	&	0.1420	\\
25/07/2017 & -		&	-		&	0.0961	&	0.1445	&	0.0961	&	0.1445	\\
26/07/2017 & -		&	-		&	0.1018	&	0.1519	&	0.1018	&	0.1519	\\
27/07/2017 & -		&	-		&	0.1028	&	0.1508	&	0.1028	&	0.1508	\\
28/07/2017 & 0.0395	&	0.0833	&	0.1535	&\textbf{0.0932}	&	0.1585	&\textbf{0.1250}	\\
29/07/2017 & 0.0275	&	0.0717	&	0.1601	&\textbf{0.0919}	&	0.1625	&\textbf{0.1165}	\\
30/07/2017 & 0.0308	&	0.0710	&	0.0649	&	0.0850	&	0.0718	&	0.1107	\\
31/07/2017 & 0.0243	&	0.0621	&	0.0677	&	0.0795	&	0.0720	&	0.1009	\\
	\bottomrule
	\end{tabular} 
\end{table}

\section{Concluding remarks}\label{sec:remarks}

In this paper we borrow the idea, suggested in recent literature, that BitCoin prices are driven by sentiment on the BitCoin system and underlying technology. 
%In particular, we believe that over-confidence may explain the Bubbles documented in several papers about the analysed cryptocurrency. 
Main references in this area are \citet{MainDrivers, GoogleTrends, SentimentAnalysis,bukovina2016sentiment}. In order to account for such behavior we develop a model in continuous time which describes the dynamics of two factors, one representing the sentiment index on the BitCoin system and the other representing the BitCoin price itself, which is directly affected by the first factor; we also take into account a delay between the sentiment index and its delivered effect on the BitCoin price. 
%A different approach is considered in \citet{HencGou}  where the authors  model BitCoin price and the presence of  market Bubbles through a non causal discrete time model.
We investigate statistical properties of the proposed model and we %give conditions under which the model is arbitrage-free. 
show its arbitrage-free property. Under our model assumption we derived a closed form approximation for the joint density of the discretely observed process and we proposed a statistical estimation for that model. By applying the classical risk-neutral evaluation we are able to derive a quasi-closed formula for European style derivatives on the BitCoin with special attention of Plain Vanilla and Binary options for which a market already exists (e.g. https://deribit.com, https://coinut.com ). 
Of course sentiment about BitCoin or, more generally, on cryptocurrencies or IT finance is not directly observed  but several variables may be considered as indicators. Here, we analyzed the volume and more unconventional sentiment indicators such as the number of Google searches and the number of Wikipedia requests about the topic (as suggested \citet{GoogleTrends, SentimentAnalysis}). 
First of all, we investigated whether these proxies were consistent with the suggested model and we proved that both the volume of transactions and the number of Google searches give a good fit of the dynamics described in the model. 
Finally we fit the model using real data of BitCoin price with Volume of transactions and Google searches respectively and we provided the estimation results.   
%As for future research, our attention will be also devoted to the full specification model with non-zero correlation between the two factors. In fact, as Figure \ref{CambioRho} \fbox{INSERIRE?} suggests, we believe that the model we introduced is capable of describing Bubbles in the BitCoin market by simply modulating the correlation parameter value. 
%Last but not least, taking advantage of the statistical properties described in Section 2, we will address the fit of the suggested model to observed data.
Several open problems are left for future research. As a first issue would like to address a multivariate extension of the model in order to take into consideration the special feature of BitCoin being traded in different exchanges and related stylized facts and to investigate whether there are arbitrage opportunities between different BitCoin exchanges. Besides we would like to investigate whether the model is suitable to describe bubble effects which have been also evidenced for the BitCoin price dynamics.

\medskip

 \begin{center}
{\bf Acknowledgments}
\end{center}
The first and the second named authors are grateful to Banca d'Italia and Fondazione Cassa di Risparmio di Perugia for the financial support.

\bibliographystyle{plainnat}
\bibliography{biblio_BS1}

\appendix

\section{Levy approximation}\label{appendix:levy}

In \citet{levy1992pricing} the author proves that the distribution of the mean integrated Brownian motion $\frac{1}{s} \int_0^s P_u \ud u$ can be approximated with a log-normal distribution, at least for suitable values of the model parameters $\mu_P, \sigma_P$; the parameters of the approximating log-normal distribution are obtained by applying a moment matching technique. 
Set
\begin{equation}\label{def:IP}
IP(s):=\int_0^s P_u \ud u,\quad s>0.
\end{equation}
Of course, the distribution of  $IP(s)$ can also be approximated by a log-normal for $s >0$. By applying the moment matching technique the parameters of the corresponding log-Normal distribution for $IP(s)$ are given by
$$
\alpha(s)=\log\left(\frac{\esp{ IP(s)}^2}{\sqrt{\esp{IP(s)^2}}}\right),
$$
$$
\nu^2(s)=\log \left(\frac{\esp{IP(s)^2}}{\esp{IP(s)}^2}\right),
$$
The approximate distribution density function of $IP(s)$ is thus given by
$$
f_{IP(s)}(x)= \mathcal{LN} pdf_{\alpha(s),\nu^2(s)} \left(x\right), \quad \mbox{ if } s > 0
$$

where $ \mathcal{LN} pdf_{m,v} $ denotes the probability distribution function of a log-normal distribution with parameters $m$ and $v$, defined as

$$
\mathcal{LN} pdf_{m,v}(y)=\frac{1}{y\sqrt{2\pi v}} e^{-\frac{(log (y)-m)^2}{2v}}, \quad \forall\ y \in \R^+.
$$

In the paper the above approximation is applied twice with completely different purposes. In Section \ref{sec:option-pricing} it is applied to derive an approximate distribution for the integrated sentiment process starting at $t=0$ i.e. to $X_{0,T}^\tau=X_\tau^\tau +IP(T-\tau)$.
Note that $X_{0,T}^\tau-X_\tau^\tau =IP(T-\tau)$ hence the derivations of its distribution is trivial once that of $IP(T-\tau)$ is known.

In Section \ref{sec:fit}, once $\tau<\Delta$ is assigned, the Levy approximation is applied to derive the distribution of $A_1$ and of $A_i$ given $A_{i-1}$ where
 $A_1=X_\tau^\tau + IP(\Delta-\tau)$ and, for $i\geq 2$, 
 $A_i=\int_{(i-1)\Delta-\tau}^{i\Delta-\tau} P_u du=\int_{-\tau}^{\Delta-\tau} P_{u+(i-1)\Delta} \ud u =P_{(i-1)\Delta}\int_{-\tau}^{\Delta-\tau} P_{u} \ud u =P_{(i-1)\Delta} \left( X_\tau^\tau +IP(\Delta-\tau) \right)$.
 
%If $(j-1)\Delta<\tau<j\Delta$ then $A_i, i\leq j-1 $ are deterministic $A_j$ replaces $A_1$.

\section{Technical proofs}\label{appendix:technical}

\begin{proof}[Proof of Lemma \ref{th:means}]
In order to prove the Lemma let us first compute the mean and the variance of $IP(s)$ given in \eqref{def:IP} for each $s>0$.

Fix $s >0$. Since $P_u > 0$ for each $u \in  (0, s]$, by applying Fubini’s theorem we get
$$
\esp{IP(s)} = \esp{\int_{0}^{s}P_{u}\ud u}=\int_{0}^{s}\esp{P_{u}}\ud u,
$$
where, for each $u \ge 0$, we have
\begin{align*}
\esp{P_{u}} & 
 = \phi(0) e^{\left(\mu_{P}-\frac{\sigma_{P}^{2}}{2}\right)u} \esp{e^{\sigma_{P}Z_{u}}}
%= \phi(0)e^{\left(\mu_{P}-\frac{\sigma_{P}^{2}}{2}\right)u} e^{\frac{\sigma_{P}^{2}u}{2}} \\
= \phi(0)e^{\mu_{P}u},
\end{align*}
since $P$ is a geometric Brownian motion with $P_0=\phi(0)$.
Hence
$$
\esp{IP(s)}=\phi(0)\int_{0}^{s} e^{\mu_{P}u} \ud u=\frac{\phi(0)}{\mu_{P}}\left(e^{\mu_{P}(s)}-1\right).%\quad s > 0.
$$
As for the variance of $IP(s)$, %for each $s > 0$, 
we have
\begin{align}
\mathbb V{\rm ar}[IP(s)] &
= \mathbb V{\rm ar}\left[\int_0^{s} P_u \ud u\right]
= \esp{\left(\int_0^{s} P_u \ud u\right)^2} - \esp{\int_0^{s} P_u \ud u}^2,
\end{align}
with
\begin{align}
\esp{\left(\int_0^{s} P_u \ud u\right)^{2}} & = 
2\esp{\int_{0}^{s}P_{v}dv\int_{0}^{v}P_{u}\ud u}
=2\esp{\int_{0}^{s}\int_{0}^{v}P_{u}P_{v}\ud v\ud u}\\
& = 2\int_{0}^{s}\int_{0}^{v}\esp{P_{u}P_{v}}\ud v\ud u, \label{eq:MP2}
\end{align}
where the last equality again holds thanks to Fubini’s theorem. Moreover, by the independence property of the increments of Brownian motion, for $0 < u < v \leq s$, we get
\begin{align*}
\esp{P_{u}P_{v}} & = \esp{P_{u}^{2}e^{\left(\mu_{P}-\frac{\sigma_{P}^{2}}{2}\right)(v-u)+\sigma_{P}\left(Z_{v}-Z_{u}\right)}}\\
& =e^{\left(\mu_{P}-\frac{\sigma_{P}^{2}}{2}\right)(v-u)} \esp{P_{u}^{2}\esp{e^{\sigma_{P}\left(Z_{v}-Z_{u}\right)}|\F_u^P}}\\
& =e^{\left(\mu_{P}-\frac{\sigma_{P}^{2}}{2}\right)(v-u)} \esp{P_{u}^{2}}\esp{e^{\sigma_{P}\left(Z_{v}-Z_{u}\right)}}\\
& =e^{\left(\mu_{P}-\frac{\sigma_{P}^{2}}{2}\right)(v-u)} \esp{P_{u}^{2}}\esp{e^{\frac{\sigma_{P}^{2}\left(v-u)\right)}{2}}}=e^{\mu_{P}(v-u)}\esp{P_{u}^{2}}.
\end{align*}
Further,
\begin{align*}
\esp{P_{u}^{2}} & =\phi^2(0)e^{2\left(\mu_{P}-\frac{\sigma_{P}^{2}}{2}\right)u}\esp{e^{2\sigma_{P}Z_{u}}}
=\phi^2(0)e^{\left(2\mu_{P}+\sigma_{P}^{2}\right)u}.
\end{align*}
Hence
\begin{equation} \label{eq:pp}
\esp{P_{u}P_{v}}=\phi^2(0)e^{\mu_{P}(v-u)}e^{\left(2\mu_{P}+\sigma_{P}^{2}\right)u},
\end{equation}
and by plugging \eqref{eq:pp} into \eqref{eq:MP2}, we have
\begin{align*}
\esp{IP(s)^{2}}
& = 2\phi^2(0)\int_{0}^{s}e^{\mu_{P}v} \int_{0}^{v}e^{\left(\mu_{P}+\sigma_{P}^{2}\right)u} \ud u\ud v\\
&  = \frac{2\phi^{2}(0)}{\left(\mu_{P}+\sigma_{P}^{2}\right)\left(2\mu_{P}+\sigma_{P}^{2}\right)}\left(e^{\left(2\mu_{P}+\sigma_{P}^{2}\right)s} -1\right)-\frac{2\phi^{2}(0)}{\mu_{P}\left(\mu_{P}+\sigma_{P}^{2}\right)}\left(e^{\mu_{P}s} -1\right).
\end{align*}
Finally, gathering the results we get
\begin{align*}
\mathbb V{\rm ar}[IP(s)] & = 
\frac{2\phi^{2}(0)}{\left(\mu_{P}+\sigma_{P}^{2}\right)\left(2\mu_{P}+\sigma_{P}^{2}\right)}\left[e^{\left(2\mu_{P}+\sigma_{P}^{2}\right)s} -1\right] \\
& \qquad - \frac{2\phi^{2}(0)}{\mu_{P}\left(\mu_{P}+\sigma_{P}^{2}\right)}\left(e^{\mu_{P}s} -1\right)-\left(\frac{\phi(0)}{\mu_{P}}\left(e^{\mu_{P} s} -1\right)\right)^2.
\end{align*}

Note that $X_t^\tau$, with $t\in [0,\tau]$, and $X_{t,T}^\tau$, with $t <T \leq \tau$,
are fully deterministic and the computation is trivial.

To prove points (i)-(iii), it suffices to observe that
$$
X_t^\tau=X_\tau^\tau+IP(t-\tau),\quad t > \tau, 
$$
$$
X_{t,T}^\tau=\int_{t-\tau}^0 \phi(u)\ud u + IP(T-\tau),\quad t \leq \tau < T, 
$$
and the computation easily follows once those of $IP(s)$ are known for $s>0$. 

To prove point (ii), it is worth noticing that given $0 \leq v<s$
\begin{align}
IP(s)-IP(v) & =\int_v^s P_u \ud u = \int_v^{s} P_v e^{\left(\mu_P-\frac{\sigma_P^2}{2}\right)(u-v)+\sigma_P(Z_u-Z_v)} \ud u\stackrel{({\rm law})}{=} P_v\int_v^{s} e^{\left(\mu_P-\frac{\sigma_P^2}{2}\right)(u-v)+\sigma_P Z_{u-v} } \ud u  \\
		& = P_v\int_0^{s-v} e^{\left(\mu_P-\frac{\sigma_P^2}{2}\right)r+\sigma_P Z_r} \ud r=P_v IP(s-v), \label{eq:u-t}
\end{align}
where $r=u-v$.

%By taking the expectation of both sides of \eqref{eq:u-t}, we have
%\begingroup
%\allowdisplaybreaks
%\begin{align}
%\mathbb{E}\left[\int_t^{T} e^{\left(\mu_P-\frac{\sigma_P^2}{2}\right)(u-t)+\sigma_P Z_{u-t}} \ud u \right] & = \mathbb{E}\left[\int_0^{T-t} e^{\left(\mu_P-\frac{\sigma_P^2}{2}\right)v+\sigma_P Z_v } \ud v \right]  \\
%	& = \int_0^{T-t} e^{\left(\mu_P-\frac{\sigma_P^2}{2}\right)v} \mathbb{E}\left[ e^{\sigma_P Z_v}\right] \ud v  \\
%	& = \int_0^{T-t} e^{\left(\mu_P-\frac{\sigma_P^2}{2}\right)v} e^{\frac{\sigma_P^2}{2}v} \ud v \\
%%	& = \int_0^{T-t} e^{\left(\mu_P-\frac{\sigma_P^2}{2}+\frac{\sigma_P^2}{2}\right)v} \ud v \\
%	& = \int_0^{T-t} e^{\mu_P v} \ud v \\
%	& = \frac{1}{\mu_P} \left(e^{\mu_P(T-t)}-1\right).
%\end{align}
%\endgroup
%Finally, we obtain
%\begin{equation}
%\mathbb{E} \left[X_{t,T}^0\right] = \mathbb{E} [P_t] \frac{e^{\mu_P(T-t)}-1}{\mu_P} = \phi(0)\, e^{\mu_P t}  \frac{e^{\mu_P(T-t)}-1}{\mu_P}.
%\end{equation}
%%

%since
%\begin{equation}
%\mathbb{E}\left[ P_t \right] =\mathbb{E}\left[ P_0 e^{\left(\mu_P \frac{\sigma_P^2}{2}\right)t+\sigma_P Z_t} \right] = P_0\, e^{\mu_p t}
%\end{equation}

To obtain the desired result it suffices to note that, for $\tau\leq t<T$,
\begin{align*}
X_{t,T}^\tau & =IP(T-\tau)-IP(t-\tau)
%\int_{t-\tau}^{T-\tau} P_u \ud u\stackrel{({\rm law})}{=}  P_t\int_{-\tau}^{T-t-\tau} e^{(\mu_P-%\frac{\sigma_P^2}{2})u+\sigma_P Z_u} \ud u.
%&= \qquad  CONTROLLA
\end{align*}
and apply \eqref{eq:u-t}. The computation of the mean and variance of the above difference is straightforward given the independence of Brownian increments.
\end{proof}

\begin{proof}[Proof of Lemma \ref{lem:measure}]

Firstly, we prove that formula \eqref{eq:L} defines a probability measure $\Q$ equivalent to $\P$ on $(\Omega,\F_T)$. This means we need to show that $L^\Q$ is an $(\bF,\P)$-martingale, that is, $\esp{L_T^\Q}=1$. Since the $\bF$-progressively measurable process $\gamma$ can be suitably chosen, to prove this relation we can assume $\gamma \equiv 0$, without loss of generality.  
Set 
\begin{equation}\label{def:alpha}
\alpha_t := -\frac{\mu_S P_{t-\tau}-r(t)}{\sigma_S \sqrt{P_{t-\tau}}}, \quad t \in [0,T].
\end{equation}
We observe that since $\phi(t) > 0$, for each $t \in [-L,0]$, in \eqref{eq:BSdyn},
by Theorem \ref{th:sol}, point (i), we have that $P_{t-\tau} > 0$, $\P$-a.s. for all $t \in [0,T]$, so that the process $\alpha=\{\alpha_t,\ t \in [0,T]\}$ given in \eqref{def:alpha} is well-defined, as well as the random variable $L_T^\Q$.
%We apply similar arguments to those used in \citet[Section 3]{arriojas2007delayed}.
Clearly, $\alpha$ is an $\bF$-progressively measurable process. Moreover, since the trajectories of the process $P$ are continuous, then $P$ is almost surely bounded on $[0,T]$ and this implies that
$\int_0^T|\alpha_u|^2 \ud u < \infty$ $\P$-a.s.;
%, since sample-path continuity of the process $P$ yields the fulfillment of the almost sure boundedness property by $P$; 
on the other hand, the condition $\phi(t)>0$, for every $t \in [-L,0]$, implies that almost every path of $\left\{\frac{1}{\sigma_S \sqrt{P_{t-\tau}}},\ t \in [0,T]\right\}$ is bounded on the compact interval $[0,T]$.
Set $\F_t^P:=\F_0^P=\{\Omega,\emptyset\}$, for $t \leq 0$. Then, $\alpha_u$, for every $u \in [0,T]$, is $\F_{T-\tau}^P$-measurable. 
Since $Z_{u-\tau}$ is independent of $W_u$, for every $u \in [\tau, T]$, the stochastic integral $\int_0^T \alpha_u \ud W_u$ conditioned on $\F_{T-\tau}^P$ has a normal distribution with mean zero and variance $\int_0^T |\alpha_u|^2 \ud u$. Consequently, the formula for the moment generating function of a normal distribution implies
$$
\esp{e^{\int_0^T \alpha_u \ud W_u}\bigg{|}\F_{T-\tau}^P}=e^{\frac{1}{2}\int_0^T |\alpha_u|^2 \ud u},
$$
or equivalently
\begin{equation}\label{eq:gen}
\esp{e^{\int_0^T \alpha_u \ud W_u-\frac{1}{2}\int_0^T |\alpha_u|^2 \ud u}\bigg{|}\F_{T-\tau}^P}=1.
\end{equation}
Taking the expectation of both sides of \eqref{eq:gen} immediately yields $\esp{L_T^\Q}=1$.
Now, set $\widetilde S_t:= \ds \frac{S_t}{B_t}$, for each $t \in [0,T]$.
It remains to verify that the discounted BitCoin price process $\widetilde S=\{\widetilde S_t,\ t \in [0,T]\}$ is an $(\bF,\Q)$-martingale.
By Girsanov's theorem, under the change of measure from $\P$ to $\Q$, we have two independent $(\bF,\Q)$-Brownian motions $W^\Q=\{W_t^\Q,\ t \in [0,T]\}$ and $Z^\Q=\{Z_t^\Q,\ t \in [0,T]\}$ defined respectively by
\begin{align*}
W_t^\Q & := W_t - \int_0^t\alpha_u \ud u,\quad t \in [0,T],\\
Z_t^\Q & := Z_t + \int_0^t\gamma_s \ud s,
\quad t \in [0,T].
\end{align*}
Under the martingale measure $\Q$, the discounted BitCoin price process $\widetilde S$ satisfies the following dynamics
\begin{align*}
%\ud \widetilde P_t & = \widetilde P_t\sigma_P\ud \widetilde Z_t\\
\ud \widetilde S_t & =  \widetilde S_t\sigma_S\sqrt{P_{t-\tau}}\ud W_t^\Q, \quad \widetilde S_0=s_0 \in \R_+,
\end{align*}
which implies that $\widetilde S$ is an $(\bF,\Q)$-local martingale. Finally, proceeding as above it is easy to check that $\widetilde S$ is a true $(\bF,\Q)$-martingale.

\end{proof}

\begin{proof}[Proof of Lemma \ref{lemma:distr}] 
By applying \citet{levy1992pricing} we have (see Appendix \ref{appendix:levy}) that the distribution of $A_1-X_\tau^\tau$ can be approximated by a log-normal with parameters
\begin{align}
\alpha_1 & =\log\left(\frac{\mathbb{E}[IP(\Delta-\tau)]^2}{\sqrt{\mathbb{E}[IP(\Delta-\tau)^2]}} \right),\qquad 
\nu_1^2 = \log\left( \frac{\mathbb{E}[IP(\Delta-\tau)^2]}{\mathbb{E}[IP(\Delta-\tau)]^2} \right).
\end{align}
By applying the outcomes of Lemma \ref{th:means}, we have 
\begin{align}
\mathbb{E}[A_1] & = \mathbb{E} [P_0] \frac{e^{\mu_P (\Delta-\tau)}-1}{\mu_P} = P_0 \frac{e^{\mu_P (\Delta-\tau)}-1}{\mu_P}\\
\mathbb{E}[A_1^2] & = \mathbb{E} [P_0^2] \left( \frac{2}{\mu_P+\sigma_P^2} \left[ \frac{e^{\left(2\mu_P+\sigma_P^2\right)(\Delta-\tau)}-1}{2\mu_P+\sigma_P^2}-\frac{e^{\mu_P(\Delta-\tau)}-1}{\mu_P}   \right] \right)\\
	& = P_0^2 \left( \frac{2}{\mu_P+\sigma_P^2} \left[ \frac{e^{\left(2\mu_P+\sigma_P^2\right)(\Delta-\tau)}-1}{2\mu_P+\sigma_P^2}-\frac{e^{\mu_P(\Delta-\tau)}-1}{\mu_P}   \right] \right).
\end{align}
Hence
\begin{align}
\alpha_1 & = \log \left(\frac{P_0\left( \frac{e^{\mu_P \Delta}-1}{\mu_P} \right)^2}{\sqrt{ \frac{2}{\mu_P+\sigma_P^2} \left[ \frac{e^{\left(2\mu_P+\sigma_P^2\right)\Delta}-1}{2\mu_P+\sigma_P^2}-\frac{e^{\mu_P\Delta}-1}{\mu_P}   \right]}}\right)  \\
\nu_1^2 & = \log\left(\frac{ \frac{2}{\mu_P+\sigma_P^2} \left[ \frac{e^{\left(2\mu_P+\sigma_P^2\right)\Delta}-1}{2\mu_P+\sigma_P^2}-\frac{e^{\mu_P\Delta}-1}{\mu_P}   \right]}{\left( \frac{e^{\mu_P \Delta}-1}{\mu_P} \right)^2}\right).
\end{align}
By applying simple computation we get the outcomes for part (i). Moreover,
\begin{align}
\mathbb{E}[A_i] & = \mathbb{E} [P_{i-1}] \frac{e^{\mu_P \Delta}-1}{\mu_P} = \mathbb{E}[P_{i-2}] \; e^{\mu_P\Delta} \; \frac{e^{\mu_P \Delta}-1}{\mu_P} = \mathbb{E}[A_{i-1}]\; e^{\mu_P\Delta}\\
\mathbb{E}[A_i^2] & = \mathbb{E} [P_{i-1}^2] \left( \frac{2}{\mu_P+\sigma_P^2} \left[ \frac{e^{\left(2\mu_P+\sigma_P^2\right)\Delta}-1}{2\mu_P+\sigma_P^2}-\frac{e^{\mu_P\Delta}-1}{\mu_P}   \right] \right)\\
	& = \mathbb{E}[P_{i-2}^2] e^{\left(2\mu_P+\sigma_P^2\right)\Delta} \left( \frac{2}{\mu_P+\sigma_P^2} \left[ \frac{e^{\left(2\mu_P+\sigma_P^2\right)\Delta}-1}{2\mu_P+\sigma_P^2}-\frac{e^{\mu_P\Delta}-1}{\mu_P}   \right] \right)\\
	& = \mathbb{E}[A_{i-1}^2]\; e^{\left(2\mu_P+\sigma_P^2\right)\Delta}.
\end{align}
Then
\begin{align}
\alpha_i & = \log\left( \frac{\mathbb{E}[A_i]^2}{\sqrt{\mathbb{E}[A_i^2]}} \right) = \log\left( \frac{\mathbb{E}[A_{i-1}]^2 \left(e^{\mu_P\Delta}\right)^2}{\sqrt{\mathbb{E}[A_{i-1}^2] \; e^{\left(2\mu_P+\sigma_P^2\right)\Delta}}} \right) \\
	& = \log\left( \frac{\mathbb{E}[A_{i-1}]^2}{\sqrt{\mathbb{E}[A_{i-1}^2]}} \right) + \log\left( \frac{\left(e^{\mu_P\Delta}\right)^2}{\sqrt{e^{\left(2\mu_P+\sigma_P^2\right)\Delta}}}\right) = \alpha_{i-1} + \left( \mu_P-\frac{\sigma_P^2}{2} \right)\Delta, \\
\nu_i^2 & = \log\left( \frac{\mathbb{E}[A_i^2]}{\sqrt{\mathbb{E}[A_i]^2}} \right) = \log\left( \frac{\mathbb{E}[A_{i-1}^2] \; e^{\left(2\mu_P+\sigma_P^2\right)\Delta}}{\mathbb{E}[A_{i-1}]^2 \; e^{(\mu_P\Delta)^2}} \right)\\
	& =\log\left( \frac{\mathbb{E}[A_{i-1}^2]}{\sqrt{\mathbb{E}[A_{i-1}]^2}} \right) + \log\left( \frac{e^{\left(2\mu_P+\sigma_P^2\right)\Delta}}{e^{(\mu_P\Delta)^2}} \right) = \nu_{i-1}^2 + \sigma_P^2\Delta.
\end{align}
Conditioning to $A_{i-1}$
\begin{align}
\alpha_i & = \log\left(A_{i-1}\right) + \left( \mu_P-\frac{\sigma_P^2}{2}\right)\Delta, \\
\nu_i^2 & = \sigma_P^2\Delta,
\end{align}
which gives part (ii).

\end{proof}

%\section{Supplementary material}\label{appendix:supplementary}
\end{document}